\pdfoutput = 1
\documentclass[11pt]{article}
\usepackage{fullpage}
\usepackage[english]{babel}
\usepackage[utf8x]{inputenc}
\usepackage[T1]{fontenc}
\usepackage{amsmath}
\usepackage{amssymb}
\usepackage{amsthm}
\usepackage{bbm}
\usepackage{bbold}
\usepackage{parskip}
\usepackage{thm-restate}
\usepackage{booktabs}
\usepackage{array}
\usepackage{xcolor}

\usepackage{algorithm}
\usepackage{algpseudocode}

\usepackage{hyperref}
\hypersetup{
	colorlinks=true,
	linkcolor=blue!70!black,
	citecolor=blue!70!black,
	urlcolor=blue!70!black
}
\usepackage{cleveref}
\usepackage{graphicx}
\usepackage{algorithm}
\usepackage[lined,boxed,ruled,norelsize,algo2e,linesnumbered]{algorithm2e}
\hypersetup{
	colorlinks=true,
	linkcolor=blue!70!black,
	citecolor=blue!70!black,
	urlcolor=blue!70!black
}

\newtheorem{theorem}{Theorem}
\newtheorem{lemma}{Lemma}

\newtheorem{fact}{Fact}

\newtheorem{definition}{Definition}
\newtheorem{corollary}{Corollary}
\newtheorem{proposition}{Proposition}

\newtheorem{problem}{Problem}

\newcommand{\defeq}{:=}

\newcommand{\norm}[1]{\left\lVert#1\right\rVert}
\newcommand{\norms}[1]{\lVert#1\rVert}
\newcommand{\normop}[1]{\left\lVert#1\right\rVert_{\textup{op}}}
\newcommand{\normf}[1]{\left\lVert#1\right\rVert_{\textup{F}}}

\newcommand{\normsf}[1]{\lVert#1\rVert_{\textup{F}}}

\newcommand{\inprod}[2]{\left\langle#1, #2\right\rangle}
\newcommand{\inprods}[2]{\langle#1, #2\rangle}
\newcommand{\eps}{\varepsilon}

\newcommand{\argmin}{\textup{argmin}} 
\newcommand{\R}{\mathbb{R}}

\newcommand{\N}{\mathbb{N}}

\newcommand{\half}{\frac{1}{2}}
\newcommand{\thalf}{\tfrac{1}{2}}
\newcommand{\E}{\mathbb{E}}
\newcommand{\proj}{\boldsymbol{\Pi}}

\newcommand{\set}{\mathcal{K}}
\newcommand{\ball}{\mathbb{B}}
\newcommand{\id}{\mathbf{I}}

\newcommand{\dd}{\textup{d}}
\definecolor{burntorange}{rgb}{0.8, 0.33, 0.0}

\newcommand{\poly}{\textup{poly}}

\newcommand{\Par}[1]{\left(#1\right)}
\newcommand{\Brack}[1]{\left[#1\right]}
\newcommand{\Brace}[1]{\left\{#1\right\}}
\newcommand{\Abs}[1]{\left|#1\right|}

\newcommand{\oracle}{\mathcal{O}}
\newcommand{\alg}{\mathsf{alg}}
\newcommand{\ind}{\mathbb{I}}

\newcommand{\ma}{\mathbf{A}}

\newcommand{\mc}{\mathbf{C}}

\newcommand{\mf}{\mathbf{F}}

\newcommand{\mm}{\mathbf{M}}

\newcommand{\va}{\mathbf{a}}
\newcommand{\vb}{\mathbf{b}}
\newcommand{\vc}{\mathbf{c}}
\newcommand{\vd}{\mathbf{d}}
\newcommand{\ve}{\mathbf{e}}
\newcommand{\vf}{\mathbf{f}}
\newcommand{\vg}{\mathbf{g}}

\newcommand{\vk}{\mathbf{k}}

\newcommand{\vp}{\mathbf{p}}
\newcommand{\vq}{\mathbf{q}}

\newcommand{\vs}{\mathbf{s}}
\newcommand{\vt}{\mathbf{t}}
\newcommand{\vu}{\mathbf{u}}
\newcommand{\vv}{\mathbf{v}}
\newcommand{\vw}{\mathbf{w}}
\newcommand{\vx}{\mathbf{x}}
\newcommand{\vy}{\mathbf{y}}

\newcommand{\0}{\mathbf{0}}
\newcommand{\1}{\mathbf{1}}
\newcommand{\bvx}{\bar{\vx}}

\newcommand{\calA}{\mathcal{A}}
\newcommand{\calB}{\mathcal{B}}
\newcommand{\calC}{\mathcal{C}}
\newcommand{\calD}{\mathcal{D}}

\newcommand{\calF}{\mathcal{F}}

\newcommand{\calH}{\mathcal{H}}

\newcommand{\calK}{\mathcal{K}}
\newcommand{\calL}{\mathcal{L}}

\newcommand{\calN}{\mathcal{N}}
\newcommand{\calO}{\mathcal{O}}
\newcommand{\calP}{\mathcal{P}}

\newcommand{\calU}{\mathcal{U}}
\newcommand{\calV}{\mathcal{V}}
\newcommand{\calW}{\mathcal{W}}
\newcommand{\calX}{\mathcal{X}}

\newcommand{\event}{\mathcal{E}}

\newcommand{\xset}{\calX}

\newcommand{\vset}{\calV}

\newcommand{\Lprop}{\calL_{\textup{prop}}}
\newcommand{\Lglm}{\calL_{\textup{GLM}}}
\newcommand{\clift}{\textup{clift}}
\newcommand{\reg}{\textup{reg}}
\newcommand{\sa}{\mathsf{SimultaneousApproach}}
\newcommand{\codeStyle}[1]{{\bfseries #1} }

\newcommand{\codeReturn}{\codeStyle{Return:}}

\newcommand{\tvg}{\tilde{\vg}}

\newcommand{\sign}{\textup{sign}}

\newcommand{\sps}[1]{^{(#1)}}

\newcommand{\Clin}{\calC_{\textup{lin}}}
\newcommand{\Wthresh}{\calW_{\textup{thresh}}}
\newcommand{\srad}{\mathsf{srad}}
\newcommand{\simu}{\sim_{\textup{unif.}}}
\newcommand{\csa}{\mathsf{ContextualSimultaneousApproach}}
\newcommand{\rad}{\mathsf{rad}}
\newcommand{\simiid}{\sim_{\textup{i.i.d.}}}
\newcommand{\TLC}{T_{\calL, \calC}}
\newcommand{\seq}{\textup{seq}}
\newcommand{\stat}{\textup{stat}}
\renewcommand{\epsilon}{\varepsilon}

\newcommand{\uni}{\textup{uni}}

\usepackage{mathtools}

\title{Simultaneous Blackwell Approachability \\ and Applications to Multiclass Omniprediction}
\author{
Lunjia Hu\thanks{Northeastern University, \texttt{lunjia@alumni.stanford.edu}} \and 
Kevin Tian\thanks{University of Texas at Austin, \texttt{kjtian@cs.utexas.edu}} \and 
Chutong Yang\thanks{University of Texas at Austin, \texttt{cyang98@utexas.edu}}
}
\date{}
\allowdisplaybreaks
\begin{document}

\maketitle
\begin{abstract}
Omniprediction is a learning problem that requires suboptimality bounds for each of a family of losses $\calL$ against a family of comparator predictors $\calC$. We initiate the study of omniprediction in a multiclass setting, where the comparator family $\calC$ may be infinite. Our main result is an extension of the recent binary omniprediction algorithm of \cite{OkoroaforKK25} to the multiclass setting, with sample complexity (in statistical settings) or regret horizon (in online settings) $\approx \eps^{-(k+1)}$, for $\eps$-omniprediction in a $k$-class prediction problem. En route to proving this result, we design a framework of potential broader interest for solving Blackwell approachability problems where multiple sets must simultaneously be approached via coupled actions.
\end{abstract}

\thispagestyle{empty}
\newpage
\tableofcontents
\thispagestyle{empty}
\newpage

\section{Introduction}\label{sec:intro}
\setcounter{page}{1}

Omniprediction is a powerful definition of learning introduced recently by \cite{GopalanKRSW22}. Consider a standard supervised learning task: we receive i.i.d.\ samples $(\vx, \vy) \sim \calD$, where $\vx \in \R^d$ are the features and $\vy \in \partial \Delta^k \defeq \{\ve_i\}_{i \in [k]}$ is the label  (see Section~\ref{ssec:notation} for notation), and we wish to build a predictor $\vp(\vx) \approx \E[\vy \mid \vx]$. In omniprediction, a family of loss functions $\calL$ is fixed, as well as a family of comparator predictors $\calC$. The
goal is then to satisfy the simultaneous loss minimization guarantee, for some $\eps > 0$ and predictor $\vp: \R^d \to \Delta^k$:
\begin{equation}\label{eq:omni_intro}\E_{(\vx, \vy) \sim \calD}\Brack{\ell\Par{\vk^\star_\ell(\vp(\vx)), \vy}} \le \min_{\vc \in \calC} \E_{(\vx, \vy) \sim \calD}\Brack{\ell\Par{\vc(\vx), \vy}}+\eps, \text{ for all } \ell \in \calL. \end{equation}
Here, $\vk^\star_\ell$ is the \emph{ex ante optimum} mapping for a particular loss $\ell \in \calL$, defined in \eqref{eq:eao}. This function maps each $\vp \in \Delta^k$ to the loss-minimizing action, on average over $\vy = \ve_i$ where $i \sim \vp$.

The formulation \eqref{eq:omni_intro} effectively decouples the tasks of \emph{prediction} and \emph{action}: once the learner has decided on a predictor $\vp$, the decision maker who wishes to minimize a particular loss $\ell \in \calL$ then takes the action $\vk^\star_\ell \circ \vp$. This property is particularly useful when e.g., losses can depend on parameters unknown at training time (such as a market price), or robustness to a range of loss hyperparameters is desirable.
Because \eqref{eq:omni_intro} applies to a family of losses, the predictor $\vp$ can be viewed as a ``supervised sufficient statistic'' that goes beyond single loss minimization. This perspective built upon earlier work in algorithmic fairness \cite{Hebert-JohnsonK18}, and has intimate connections to indistinguishability arguments from pseudorandomness \cite{GopalanHKRW23, GopalanH25}.

By now, there is a rich body of work on omniprediction in statistical and online learning settings \cite{GopalanKRSW22, GopalanHKRW23, HuNRY23, GopalanKR23, GargJRR24, hu2025omnipredicting, DworkHIPT25, OkoroaforKK25}. However, essentially all prior works focused on binary classification, where labels live in the set $\{0, 1\}$. This is a rather stringent restriction in the context of real-world supervised learning, which is often used for multiclass tasks, e.g., \cite{DengDSLLL09, Maas11, Deng12}. Even the ability to handle labels $\vy \in \partial \Delta^k \equiv [k]$, for $k$ a constant number of classes, would substantially extend the applicability of omnipredictors.

To our knowledge, the problem of multiclass omniprediction has only been studied in recent works by \cite{NoarovRRX25, LuRS25}. These papers focused on a setting motivated by the economics literature, where $\calC$ the family of comparators (viewed as an action space) is finite. The former's main multiclass omniprediction result (Theorem 6.5, \cite{NoarovRRX25}) is restricted to $\ell$ that independently decompose coordinatewise. On the other hand, Corollary 6, \cite{LuRS25} gives a more general statement for multiclass omniprediction, but again the result is stated for finite $\calC$, and incurs an $\approx \eps^{-4k-2}$ overhead in the sample complexity for achieving \eqref{eq:omni_intro} (without the consideration of runtime).

The main motivation of our work is to bridge this gap, by developing multiclass omnipredictors with guarantees more closely resembling the state-of-the-art in binary omniprediction. Indeed, there has been substantial recent progress on improving the sample complexity and runtime of binary omniprediction for concrete pairs $(\calC, \calL)$. For example, in the \emph{generalized linear model} (GLM) setting, where $\calC$ is bounded linear predictors and $\calL$ is appropriate convex losses (cf.\ \eqref{eq:glm_def}), \cite{hu2025omnipredicting, OkoroaforKK25} developed end-to-end efficient algorithms with $\approx \eps^{-2}$ sample complexities.\footnote{The \cite{hu2025omnipredicting} omnipredictor requires $\calL$ to be well-conditioned, but outputs a mixture of proper hypotheses from $\calC$; the \cite{OkoroaforKK25} omnipredictor is improper but holds for general bounded losses. The sample complexity of $\approx \eps^{-2}$ is tight for GLMs, even for a single loss \cite{Shamir15}. See also \cite{DworkHIPT25}, who gave a similar result in a RKHS setting.} In fact, \cite{OkoroaforKK25} gave a substantial generalization, showing how to reduce binary omniprediction for arbitrary pairs of $(\calC, \calL)$ to online learning tasks against appropriate function classes.

\subsection{Our results}\label{ssec:results}

Our approach to multiclass omniprediction is based on the framework of \cite{OkoroaforKK25}. Both \cite{hu2025omnipredicting, OkoroaforKK25}, as well as many prior results on binary omniprediction, leverage a reduction from \cite{GopalanHKRW23}. This reduction (Proposition~\ref{prop:omni}) shows that \eqref{eq:omni_intro} is satisfied for predictors $\vp$ satisfying appropriate notions of \emph{multiaccuracy} (Definition~\ref{def:ma}) and \emph{calibration} (Definiton~\ref{def:cal}), concepts we review in Section~\ref{ssec:omni_def}. Intuitively, these properties guarantee that our predictor $\vp(\vx)$ passes certain statistical tests against the ground truth $\vp^\star(\vx) \defeq \E[\vy \mid \vx]$, induced by the particular pair $(\calC, \calL)$ of interest.

As in the binary case, learning multiclass predictors that satisfy multiaccuracy and calibration individually is well-studied. We discuss the former in Section~\ref{ssec:multi_general}, and the latter is possible in $\approx \eps^{-(k + 1)}$ timesteps (in the online setting) and samples (in the statistical setting), as shown by seminal work of \cite{FosterV98} (see also \cite{MannorS10}). However, it is less clear how to achieve both simultaneously. 

In the binary setting, \cite{OkoroaforKK25} leveraged an existing calibration algorithm from \cite{AbernethyBH11} based on \emph{Blackwell approachability}, and augmented it to also guarantee multiaccuracy. Their analysis used several important facts about binary losses, e.g., existence of an approximate basis for proper losses (Lemma~\ref{lem:thresh_suffice}), and a custom ``halfspace satisfiability oracle'' specialized to their application (Algorithm~\ref{alg:oracle_binary}). Unfortunately, the natural extension of these tools to $k > 2$ classes both provably fail, necessitating a stronger framework capable of handling the multiclass setting.

\paragraph{Simultaneous Blackwell approachability.} Our starting point is to isolate a key technical primitive needed in the \cite{OkoroaforKK25} algorithm, and study sufficient conditions for it in greater generality. This is the focus of Section~\ref{sec:meta}; here, we provide a brief overview of the technique.

The standard setting of Blackwell approachability (reviewed in Section~\ref{ssec:blackwell}) generalizes von Neumann's minimax theorem to vector-valued games. Consider a bilinear, vector-valued function $\vv: \calA \times \calB$, and a set $\calV$, living in the same space $\calH$. Unlike the scalar setting, the following are not equivalent: for all $\vb \in \calB$ there exists $\va \in \calA$ such that $\vv(\va, \vb) \in \calV$ (``response satisfiability''), and there exists $\va \in \calA$ such that for all $\vb \in \calB$, $\vv(\va, \vb) \in \calV$ (``satisfiability''). Blackwell approachability \cite{Blackwell56} is an elegant compromise: whenever response satisfiability holds, we can choose $\{\va_t\}_{t \in [T]}$ in an online manner (before the corresponding $\{\vb_t\}_{t \in [T]}$ is revealed), such that $\lim_{T \to \infty} \frac 1 T \sum_{t \in [T]} \vv(\va_t, \vb_t) \to \calV$.
This strategy has intimate connections to calibration: since \cite{Foster99} many researchers have used ideas from approachability to design calibration algorithms. 

In Problem~\ref{prob:sba}, we propose a \emph{simultaneous} variant of Blackwell's approachability problem, where there are $m$ pairs of vector-valued functions $\vv\sps i$ and sets $\calV\sps i$. The goal is to choose a sequence of $\{\va_t\}_{t \in [T]}$ (responding online to $\{\vb_t\}_{t \in [T]}$) such that $\lim_{T \to \infty} \frac 1 T \sum_{t \in [T]} \vv\sps i(\va_t, \vb_t) \to \calV\sps i$, simultaneously for all $i \in [m]$. This primitive has clear connections to omniprediction, as both (binary and multiclass) multiaccuracy and calibration can be written in the language of Blackwell approachability.

Simultaneous Blackwell approachability can naturally be cast as a (standard) Blackwell approachability instance, by lifting the vectors and sets into a product space $\calH\sps 1 \times \calH\sps 2 \times \ldots \times \calH\sps m$. However, we find the perspective in Problem~\ref{prob:sba} useful, as the sufficient condition of response satisfiability does not lift cleanly. We show in Lemma~\ref{lem:multi_harder} that even when $m = 2$, there are two one-dimensional subsets $\calV\sps 1$, $\calV\sps 2$ and corresponding vector-valued functions, that are both response satisfiable (and hence approachable in isolation), but not simultaneously approachable. 

Our main contribution in Section~\ref{sec:meta} is a sufficient condition for simultaneous Blackwell approachability, stated in the form of an oracle requirement (Definition~\ref{def:mloo}). Our oracle is natural in the context of \cite{Blackwell56}, who gave an alternate characterization of approachability: every halfspace containing $\calU$ should be satisfiable. This statement was later made algorithmic by \cite{AbernethyBH11} using online learning techniques. Our sufficient condition can then be cleanly stated as: for any halfspaces each containing one $\calU\sps i$, and any specified convex combination $\vw \in \Delta^m$, the halfspaces should be satisfiable \emph{on average} (with respect to $\vw$). We further build upon \cite{AbernethyBH11} to leverage existence of such an oracle to solve simultaneous Blackwell approachability with an explicit rate (Theorem~\ref{thm:sba}).

While our reduction is a relatively straightforward extension of \cite{AbernethyBH11} (and indeed, \cite{OkoroaforKK25} also implicitly gave a variant of Theorem~\ref{thm:sba}), we believe that explicitly isolating this sufficient condition will prove useful to the community. To ease applications, we show that Theorem~\ref{thm:sba} holds in much greater generality, including in statistical and contextual settings. We provide a high-probability guarantee capable of flexibly handling these extensions in Corollary~\ref{cor:context-sba}.

\paragraph{Multiclass omniprediction.} Our simultaneous Blackwell approachability framework reduces omniprediction to implementing an appropriate mixture linear optimization oracle (MLOO, Definition~\ref{def:mloo}), and to designing appropriate online learners for each of two sets $\calV\sps i$ in isolation (corresponding to calibrated and multiaccurate predictors). While an appropriate MLOO was explicitly given in \cite{OkoroaforKK25} for the binary omniprediction setting, it is unclear how to generalize their strategy (based on an algorithmic Sperner's lemma) to hold in higher dimensions.

Towards leveraging our results for omniprediction, in Section~\ref{ssec:mixture_oracle}, we give a meta-result for designing MLOOs when all $\vv\sps i$ share a common structure. Roughly speaking, we require each $\vv\sps i$ to take as input a prediction $\vp$ and a label $\vy$, and to be linear in the prediction error $\vp - \vy$ (a more formal statement is in \eqref{eq:vvi}). Under these assumptions, we show how to use the minimax theorem and linear programming to generically design MLOOs compatible with our simultaneous Blackwell approachability framework, which extends to future potential applications.

By combining our framework, our new MLOO construction, and known online learners, we obtain our multiclass omnipredictors in Section~\ref{sec:okk}. The following is a representative result.

\begin{theorem}[Informal, see Theorem~\ref{thm:glm_multi}]\label{thm:glm_multi_intro}
Let $\calL$ be the family of multiclass GLM losses \eqref{eq:glm_def} and let $\calC$ be the family of bounded $k \times d$ linear classifiers \eqref{eq:clin_multi}. Then given $T$ i.i.d.\ samples $(\vx, \vy) \sim \calD$ for
\begin{equation}\label{eq:glm_multi_T}T = k \cdot \Omega\Par{\frac 1 \eps}^{k + 1},\end{equation}
we return an $\eps$-omnipredictor in time $O(dkT) + O(\frac 1 \eps)^{2k}\poly(k, \log \frac 1 \eps)$, with high probability.
\end{theorem}

We pause to make some remarks about Theorem~\ref{thm:glm_multi}, which is specialized to the benchmark class of multiclass GLM losses, an expressive family that includes all proper losses after reparameterization (cf.\ Lemma~\ref{lem:glm_proper}), including popular choices in practice such as the squared and cross entropy losses. First, it is fully explicit and does not rely on any computationally-infeasible oracles. Second, although its sample complexity scales exponentially in the number of classes $k$, this growth is relatively mild for small constant $k$, and the bound is independent of the ambient dimension $d$ of the features $\vx$. Third, because our approach is based on the indistinguishability argument of \cite{GopalanHKRW23} (i.e., it goes through calibration and multiaccuracy), the exponential dependence on $k$ is inevitable due to a lower bound from Theorem 1.12, \cite{HV25}. Indeed, our bound \eqref{eq:glm_multi_T} recovers the same dependence on $k$ as existing algorithms for the simpler task of multiclass calibration,\footnote{Recent works by \cite{Peng25, FishelsonGMS25} on multiclass omniprediction have traded off the exponential dependence on $k$ for an exponential dependence on $\frac 1 \eps$. We discuss these works in greater detail in Section~\ref{ssec:related}.} and improves by a quartic factor over the prior work \cite{LuRS25} (while also handling infinite $\calC$).

We prove our formal variant of Theorem~\ref{thm:glm_multi_intro} in Section~\ref{sec:okk}, as well as extensions to online omniprediction, and general families of multiclass losses and comparators (Theorem~\ref{thm:gen_multi}).

\paragraph{Other consequences.} As a warmup to our multiclass results, in Section~\ref{sec:binary}, we rederive the main results of \cite{OkoroaforKK25} in the binary setting by way of our new formalism. We believe this may be useful to the community, as it cleanly separates out the requirements of each online learner. For example, Theorem~\ref{thm:glm_bi}, our specialization of Theorem~\ref{thm:glm_multi_intro} to binary omniprediction, uses $\approx \frac 1 {\eps^2}$ samples and gives an end-to-end construction of an omnipredictor for binary GLMs. This removes the well-conditioning requirement from \cite{hu2025omnipredicting}, and does not rely on computationally-infeasible halfspace optimization oracles (i.e., ERM for linear thresholds) as required by \cite{OkoroaforKK25}.\footnote{This requirement is stated in Theorem 5, \cite{OkoroaforKK25}, where ERM access for the composition of thresholding with the comparator family is assumed. Even when the comparator family is linear functions, this requires implementing a halfspace ERM oracle, a well-known NP-hard problem in computational learning theory \cite{JohnsonP78, BenDEL03}.} During the preparation of this manuscript,  the third arXiv version of \cite{OkoroaforKK25} independently noted this oracle requirement is removable (see their updated Theorem 7.1); our modular framework makes this point transparent, which we believe will prove useful in similar future applications. 

Our framework has additional implications for the theory of calibration and omniprediction. For example, in Section~\ref{sec:union}, we show that our construction directly extends to omnipredicting against \emph{unions of comparators}, i.e., the best comparator in any of $m$ families $\{\calC\sps i\}_{i \in [m]}$, as long as we can omnipredict against each family individually. This simple extension was previously unknown, and is made possible by the generality of our construction in Section~\ref{ssec:mixture_oracle}. We are optimistic that our pipeline for constructing omnipredictors will have future consequences for related problems.

\subsection{Related work}\label{ssec:related}

\paragraph{Multiclass calibration.} Multiclass calibration has seen a resurgence of interest recently due to its use in evaluating modern classifiers in machine learning \cite{GuoPSW17}. A range of works have proposed new algorithms and relaxations of this problem \cite{KullF15, KullPKFSF19, ZhaoKSME21, GopalanHR24}. 

Of particular note, recent works~\cite{Peng25, FishelsonGMS25} gave algorithms with horizons $\approx k^{\poly(\eps^{-1})}$ for $\eps$-multiclass calibration, which is polynomial in $k$ for constant $\eps$. This improves upon the classical $\eps^{-(k + 1)}$ rate for multiclass calibration \cite{FosterV98} (as in Theorem~\ref{thm:glm_multi_intro}) in some parameter regimes. However, these results are not obtained through Blackwell approachability, and thus it seems difficult to incorporate a multiaccuracy component directly as would be required for omniprediction using the indistinguishability framework of \cite{GopalanHKRW23}. Further, the $\exp(\Omega(k))$ lower bound in Theorem 1.12, \cite{HV25} effectively rules out the use of these results within the framework.

Finally, as discussed earlier, \cite{NoarovRRX25, LuRS25} are the primary works that have studied multiclass omniprediction; we compared our Theorem~\ref{thm:glm_multi_intro} to \cite{LuRS25} earlier. Regarding \cite{NoarovRRX25}, their omniprediction result only applies to a restricted family of multiclass GLM losses, namely those which decompose coordinatewise in their argument (see Definition 6.12). This makes it incompatible with several common GLM losses in the (standard) setting of Theorem~\ref{thm:glm_multi_intro}. For example, consider the \emph{cross entropy} loss $\ell(\vp, \vy) = -\E_{i \sim \vy}[\log \vp_i]$ popularly used in machine learning evaluations. This GLM loss falls into the framework of \cite{LuRS25}, but once it is cast as a GLM learning problem, the correct parameterization is in the unlinked space (where linear comparator predictors $\vx \to \mc \vx$ live), for which the corresponding loss is $\ell(\vt, \vy) = \log(\sum_{i \in [d]}\exp(\vt_i)) - \inprod{\vt}{\vy}$, which is not coordinatewise separable. For more discussion on this point, see Lemma~\ref{lem:glm_proper} and Section 2.2, \cite{hu2025omnipredicting}.

\paragraph{Multiclass learning.} Multiclass learning is a well-studied topic in learning theory in general. For example, classical works by
\cite{Natarajan89, BenCL92} proposed various statistical dimensions characterizing the sample complexity of multiclass PAC learning. More recently,
\cite{DanielySBS15, BrukhimCDMY22} shows that in multiclass setting, empirical risk minimization does not provide a uniform bound on sample complexity, and proved that a quantity known as the DS dimension does tightly characterizes multiclass PAC learnability. 
A follow-up work by \cite{CharikarP23} shows that $k$-DS dimension, a generalization of DS dimension, characterizes $k$-list learnability. Most of these works primarily consider the sample complexity of multiclass learning rather than end-to-end efficient algorithms.

\paragraph{Blackwell approachability.} Various works have generalized Blackwell approachability and applied it to problems with a similar spirit to our simultaneous approachability setting in Problem~\ref{prob:sba}. However, to our knowledge none of them directly study achievability and algorithms for approaching multiple sets. The works most closely-related to our setup include 
\cite{MannorPS14} who study Blackwell approachability for unknown games, i.e., where the structure of the game or target is unknown;
\cite{FournierKMSW21} who study Blackwell approachability with additional constraints for the payoff space; and
\cite{LeeNPR22} who study multiclass calibration and calibeating.

\section{Preliminaries}\label{sec:prelims}

In Section~\ref{ssec:notation}, we define notation used throughout the paper, and state helper results from the online learning literature. We then introduce preliminaries for omniprediction in Section~\ref{ssec:omni_def}.

\subsection{Notation}\label{ssec:notation}

We denote vectors in lowercase boldface and matrices in uppercase boldface. For $n \in \N$ we let $[n] \defeq \{i \in \N \mid i \le n\}$. We let $\1_d$ and $\0_d$ denote the all-ones and all-zeroes vectors in $\R^d$. For $i \in [d]$, we let $\ve_i \in \R^d$ denote the $i^{\text{th}}$ standard basis vector when the dimension $d$ is clear from context. When $\event$ is some event, we let $\ind_{\event}$ to denote the $0$-$1$ indicator variable of the event.

When $\norm{\cdot}$ is a norm on $\R^d$, we let $\norm{\cdot}_*$ denote its dual norm.
For $p \in \R_{\ge 1} \cup \{\infty\}$ we let $\norm{\cdot}_p$ denote the $\ell_p$ norm of a vector argument. Note that when $\norm{\cdot} = \norm{\cdot}_p$ for some $p \in \R_{\ge 1} \cup \{\infty\}$, then $\norm{\cdot}_* = \norm{\cdot}_q$ for the value of $q \in \R_{\ge 1} \cup \{\infty\}$ satisfying $\frac 1 p + \frac 1 q = 1$.
We say that a vector-valued function $\vv: \R^d \to \R^k$ is $\beta$-Lipschitz in $\norm{\cdot}$ if for all $\vx, \vy \in \R^d$, $\norm{\vv(\vx) - \vv(\vy)}_* \le \beta\norm{\vx - \vy}$.

For $\bvx \in \R^d$ and $r > 0$ we define $\ball_p^d(\bvx, r) \defeq \{\vx \in \R^d \mid \norm{\vx - \bvx}_p \le r\}$ to be an $\ell_p$ ball centered at $\bvx$. When $\bvx$ is omitted, $\bvx = \0_d$ by default. For a compact set $\set \subseteq \R^d$ we let $\proj_\set(\vv) \defeq \arg\min_{\vx \in \set}\norm{\vv - \vx}_2$ denote the Euclidean projection.
For $p, q \ge 1$ and $\mm \in \R^{n \times d}$, we denote
\[\norm{\mm}_{p \to q} \defeq \max_{\vv \in \ball^d_p(1)} \norm{\mm \vv}_q.\]
We let $\Delta^k \defeq \{\vv \in \R^k_{\ge 0} \mid \norm{\vv}_1 = 1\}$ denote the probability simplex in dimension $k$. When $S$ is a set, we overload notation and let $\vv \in \R^S$ be a vector with coordinates indexed by elements $s \in S$, and we similarly define $\1_S$, $\0_S$, $\Delta^S$, etc. For a convex set $\set \subseteq \R^d$ we use $\partial \set$ to denote the \emph{boundary} of $\set$, i.e., all $\vv \in \set$ that cannot be written as a convex combination of other points in $\set$. For example, $\partial \Delta^k$ is the set of standard basis vectors $\{\ve_i\}_{i \in [k]}$. For a distribution $\calP$ supported on some set $\Omega$ we write $\omega \sim \calP$ to mean a sample from the distribution, and when $\vp \in \Delta^k$ we overload notation and let $i \sim \vp$ (resp.\ $\vy \sim \vp$) mean a sample that takes on the value $i$ (resp.\ $\vy = \ve_i$) with probability $\vp_i$.

We say that $\calN$ is an \emph{$\eps$-net in $\norm{\cdot}$} for $\calK$ if for all $\vx \in \calK$, there exists $\vx' \in \calN$ such that $\norm{\vx - \vx'} \le \eps$. When $\norm{\cdot}$ is omitted, we let $\norm{\cdot} = \norm{\cdot}_1$ by default. The following construction is standard.\footnote{This result is stated in \cite{Vershynin18}, Corollary 4.2.11 for the $\ell_2$ ball, but the same construction works for the $\ell_1$ ball in dimension $k - 1$. Projection onto the subspace $\1_k^\top \vv = 1$ at most doubles the $\ell_1$ distance, so we adjusted the constant.}

\begin{fact}\label{fact:net_l1}
For all $k \in \N$ and $\eps \in (0, 1)$, there exists $\calN$, an $\eps$-net of $\Delta^k$, satisfying $|\calN| \le (\frac 5 \eps)^{k - 1}$.
\end{fact}

We say that $\calH \subseteq \R^d$ is a \emph{halfspace} if for some $(\vv, c) \in \R^d \times \R$, $\calH = \{\vx \mid \vv \cdot \vx \le c\}$. For a sequence $\{\vv_t\}_{t \in \N} \subseteq \R^d$ and a set $\calV \subseteq \R^d$, we write $\lim_{t \to \infty} \vv_t \to \calV$ to mean that for any $\eps > 0$, there is some $T(\eps)$ such that for all $t \ge T(\eps)$, $\vv_t$ is within $\eps$ of the set $\calV$ in Euclidean distance.\footnote{In $\R^d$, all norms are equivalent up to universal constants, so using Euclidean distance is without loss of generality.}

We often refer to sequences of vectors, e.g., $\{\vx_t\}_{t \in [T]}$, by indexing the set of indices, e.g., as $\vx_{[T]}$. We use $f \circ g$ to denote the composition of two functions $f$ and $g$.

To instantiate our framework, we require various online learning algorithms from the literature. We begin by stating a general fact about the generalization of stochastic mirror descent methods.\footnote{Lemma 9 of \cite{hu2025omnipredicting} only claims this result for an $\ell_2$ setup, but the same argument extends to all stochastic mirror descent setups as the result simply bounds the random error term via martingale concentration.}

\begin{restatable}[Lemma 9, \cite{hu2025omnipredicting}]{lemma}{restatehighprobsgd}\label{lem:high_prob_sgd}
Let $T \in \N$, $\eta > 0$, $\delta \in (0, 1)$, let $\xset \subseteq \R^d$ have diameter $\le R$ in $\norm{\cdot}$. Let $r: \xset \to \R$ be $1$-strongly convex in a norm $\norm{\cdot}$ such that $\max_{\vx \in \xset} r(\vx) - \min_{\vx \in \xset} r(\vx) \le \Theta$, and let $\vx_1 \defeq \arg\min_{\vx \in \xset} r(\vx)$. For a sequence of deterministic vectors $\{\vg_t\}_{t\in[T]}$ such that $\vg_t$ can depend on all randomness used in iterations $t\in[T]$, let
\[\vx_{t + 1} \gets \arg\min_{\vx \in \xset}\Brace{\inprod{\eta \tvg_t - \nabla r(\vx_t)}{\vx} + r(\vx)}, \text{ where } \E\Brack{\tvg_t \mid \tvg_1, \ldots, \tvg_{t - 1}} = \vg_t, \text{ for all } t \in [T].\]
Further suppose $\norm{\tvg_t}_* \le L$ deterministically. Then for some choice of $\eta$, with probability $\ge 1 - \delta$,
\[\sup_{\vx \in \xset} \sum_{t \in [T]} \inprod{\vg_t}{\vx_t - \vx} \le 4L\sqrt{T\Theta} + 16LR\sqrt{T\log\Par{\frac 2 \delta}}.\]
\end{restatable}

We will apply two specializations of Lemma~\ref{lem:high_prob_sgd}, where $r$ is either a Euclidean regularizer (projected gradient descent), or negative entropy over the probability simplex (multiplicative weights). Finally, for deterministic applications of Lemma~\ref{lem:high_prob_sgd}, we state the following sharper bound.

\begin{lemma}[Theorem 4.2, \cite{Bubeck15}]\label{lem:mirror}
In the setting of Lemma~\ref{lem:high_prob_sgd}, if $\tvg_t = \vg_t$ in every iteration, 
\[\sup_{\vx \in \calX} \sum_{t \in [T]} \inprod{\vg_t}{\vx_t - \vx} \le L\sqrt{2T\Theta}.\]
\end{lemma}

\subsection{Omniprediction}\label{ssec:omni_def}

We consider two \emph{supervised learning} problem settings in $k$-class prediction. In the following discussion, let $\ell: \Omega \times \partial \Delta^k \to \R$ be a loss function that evaluates predictions (in a set $\Omega$) and labels.

\textbf{Online setting.} There is a sequence of examples $\{(\vx_t, \vy_t)\}_{t \in [T]}$ presented to us in an online fashion. Our goal is to predict $\{\vp_t \in \Omega\}_{t \in [T]}$ where $\vp_t$ can only depend on previous examples $\{(\vx_s, \vy_s)\}_{s < t}$ and the current features $\vx_t$, in a way that approximately minimizes $\sum_{t \in [T]} \ell(\vp_t, \vy_t)$.

\textbf{Statistical setting.} There is a distribution $\calD$ over $(\vx, \vy) \in \R^d \times \partial \Delta^k$. We refer to a pair $(\vx, \vy)$ as an \emph{example}, and we refer to the first marginal $\vx \in \R^d$ as the \emph{features} (distributed $\sim \calD_{\vx}$) and $\vy \in \partial \Delta^k$ as the \emph{label} of the example. A label $\vy = \ve_i$ represents that the example belongs to class $i \in [k]$. For example, when $k = 2$ this is the setting of binary classification. Our goal is to learn a predictor $\vp: \R^d \to \Omega$ that approximately minimizes the population loss, $\E_{(\vx, \vy) \sim \calD}[\ell(\vp(\vx), \vy)]$.

Two natural questions arise from these problem formulations: what loss $\ell$ should we consider, and what benchmark should we use to measure approximate optimality? The recently-introduced notion of omniprediction \cite{GopalanKRSW22} captures both facets of the problem simultaneously. 

Fix a family of loss functions $\calL$ such that each $\ell \in \calL$ sends $\Omega \times \partial \Delta^k \to \R$, and fix a family of \emph{comparator} predictors $\calC$ such that each $\vc \in \calC$ sends features $\vx \in \R^d$ to a prediction $\vc(\vx) \in \Omega$. We also define, for each $\ell$, an associated \emph{ex ante optimum} mapping $\vk^\star_\ell: \Delta^k \to \Omega$,
\begin{equation}\label{eq:eao}
\vk^\star_\ell(\vp) \defeq \arg\min_{\vk^\star \in \Omega}\E_{i \sim \vp}\Brack{\ell(\vk^\star, \ve_i)},
\end{equation}
with arbitrary tie-breaking. To interpret \eqref{eq:eao}, fix some distribution over labels $\vp \in \Delta^k$. Then $\vk^\star_\ell$ maps $\vp$ to the best possible prediction (according to $\ell$), had labels actually been generated $\sim \vp$.

A particularly desirable set of losses $\ell$ is those which permit taking $\vk^\star_\ell(\vp) = \vp$, i.e., where the best post-processing is just the identity function (in this case, the first argument of $\ell$ lives in $\Omega = \Delta^k$). Such losses are called \emph{proper}, and we denote the family of all proper loss functions by $\Lprop$. 

We are now ready to define omniprediction in the online and statistical settings.

\begin{definition}[Omniprediction]\label{def:omni}
In the online setting, we call $\vp_{[T]} \in (\Delta^k)^T$ an $\eps$-omnipredictor for $(\vx_{[T]}, \vy_{[T]}, \calL, \calC)$, where $\calL$ is a family of loss functions, and $\calC$ is a family of comparator predictors, if
\[\frac 1 T \sum_{t \in [T]} \ell(\vk^\star_\ell(\vp_t), \vy_t) \le \frac 1 T \sum_{t \in [T]} \ell(\vc(\vx_t), \vy_t)+ \eps, \text{ for all } \ell \in \calL,\; \vc \in \calC.\]
In the statistical setting, we call $\vp: \R^d \to \Delta^k$ an $\eps$-omnipredictor for $(\calD, \calL, \calC)$, where $\calD$ is a distribution over $\R^d \times \partial \Delta^k$, if 
\begin{align*}
\E_{(\vx, \vy) \sim \calD}\Brack{\ell(\vk^\star_\ell(\vp(\vx)), \vy)} \le \E_{(\vx, \vy) \sim \calD}\Brack{\ell(\vc(\vx), \vy)} + \eps, \text{ for all } \ell \in \calL,\; \vc \in \calC.
\end{align*}
If $\vp$ is a randomized function from $\R^d \to \Delta^k$, then we call it an $\eps$-omnipredictor for $(\calD, \calL, \calC)$ if the above display holds taking expectations over $\vp$ as well.
\end{definition}

We design omnipredictors via a recipe based on indistinguishability, pioneered by \cite{GopalanHKRW23}.

\begin{definition}[$\calF$-multiaccuracy]\label{def:ma}
Let $\calF$ be a family of functions $\vf: \R^d \to \Omega$ where $\Omega \subseteq \R^k$. In the online setting, we say that $\vp_{[T]} \in (\Delta^k)^T$  satisfies $\eps$-$(\vx_{[T]}, \vy_{[T]}, \calF)$-multiaccuracy if
\[\frac 1 T \sum_{t \in [T]} \inprod{\vp_t - \vy_t}{\vf(\vx_t)}\le \eps,\text{ for all } \vf \in \calF.\]
In the statistical setting, we say that $\vp: \R^d \to \Delta^k$ satisfies $\eps$-$(\calD, \calF)$-multiaccuracy if
\[\E_{(\vx, \vy) \sim \calD}\Brack{\inprod{\vp(\vx) - \vy}{\vf(\vx)}} \le \eps, \text{ for all } \vf \in \calF.\]
\end{definition}

\begin{definition}[$\calW$-calibration]\label{def:cal}
Let $\calW$ be a family of weight functions $\vw: \Delta^k \to \R^k$. In the online setting, we say that $\vp_{[T]} \in (\Delta^k)^T$  satisfies $\eps$-$(\vy_{[T]}, \calW)$-calibration if
\[\frac 1 T \sum_{t \in [T]} \inprod{\vp_t - \vy_t}{\vw(\vp_t)}\le \eps,\text{ for all } \vw \in \calW.\]
In the statistical setting, we say that $\vp: \R^d \to \Delta^k$ satisfies $\eps$-$(\calD, \calW)$-calibration if
\[\E_{(\vx, \vy) \sim \calD}\Brack{\inprod{\vp(\vx) - \vy}{\vw(\vp(\vx))}}, \text{ for all } \vw \in \calW.\]
\end{definition}

To connect Definitions~\ref{def:ma} and~\ref{def:cal} to omniprediction (Definition~\ref{def:omni}), the following equivalence is helpful.

\begin{lemma}\label{lem:loss_oi}
Let $\ell: \Omega \times \partial \Delta^k$ be a loss function, and define its \emph{discrete derivative} $\vd_\ell: \Omega \to \R^k$:
\begin{equation}\label{eq:dl_def}
[\vd_\ell(\vt)]_i \defeq \ell(\vt, \ve_i) \text{ for all } i \in [k].
\end{equation}
Then for any $\vt \in \Omega$, $\vp \in \Delta^k$, and $\vq \in \Delta^k$, $\E_{\vy \sim \vp}\Brack{\ell(\vt, \vy)} - \E_{\vy \sim \vq}\Brack{\ell(\vt, \vy)} = \inprod{\vd_\ell(\vt)}{\vp - \vq}$. Moreover, the same is true if we redefine $\vd_\ell(\vt) \gets \vd_\ell(\vt) - \alpha(\vt) \1_k$ for any $\alpha(\vt) \in \R$.
\end{lemma}
\begin{proof}
It suffices to expand definitions, e.g., $\E_{\vy \sim \vp}[\ell(\vt, \vy)] = \inprod{\vd_\ell(\vt)}{\vp}$, and $\1_k^\top \vp = \1_k^\top \vq = 1$.
\end{proof}

One notable loss family is $\Lglm$, the family of \emph{generalized linear model} (GLM) losses:
\begin{equation}\label{eq:glm_def}
\Lglm \defeq \Brace{\ell: \R^k \times \partial \Delta^k \to \R \mid \ell(\vt, \vy) = \omega(\vt) - \inprod{\vt}{\vy},\; \omega: \R^k \to \R \text{ convex with } \nabla \omega: \R^k \to \Delta^k}.
\end{equation}
This is because by taking $\alpha(\vt) = \omega(\vt)$ in Lemma~\ref{lem:loss_oi}, we may choose
\begin{equation}\label{eq:glm_dl}
\vd_\ell(\vt) = -\vt, \text{ for all } \ell \in \Lglm.
\end{equation}
A famous result of \cite{GneitingR07} shows that $\Lglm$ and $\Lprop$ are equivalent up to reparameterization.

\begin{lemma}[Theorem 1, \cite{GneitingR07}]\label{lem:glm_proper}
Let $\ell: \Delta^k \times \partial \Delta^k \to \R$ be a loss function. Then $\ell \in \Lprop$ iff there exists a convex function $\psi: \Delta^k \to \R$, such that $\ell(\vp, \vy) = -\psi(\vp) + \inprod{\partial \psi(\vp)}{\vp - \vy}$. Taking $\omega \defeq \psi^*$, we have that $\ell(\partial \psi(\vp), \vy) = \omega(\partial \psi(\vp)) - \inprod{\partial \psi(\vp)}{\vy}$ is a GLM loss in the argument $\vt \defeq \partial \psi(\vp)$.
\end{lemma}

We now state our omnipredictor recipe, extending \cite{GopalanHKRW23} to the multiclass setting.

\begin{restatable}{proposition}{restateomni}\label{prop:omni}
Following notation from Definition~\ref{def:omni} and \eqref{eq:dl_def}, in the online setting, if $\vp_{[T]}$ satisfies $\eps_1$-$(\vx_{[T]}, \vy_{[T]}, \calF)$-multiaccuracy and $\eps_2$-$(\vy_{[T]}, \calW)$-calibration for
\[\calF \defeq \Brace{\vd_\ell \circ \vc}_{\ell \in \calL, \vc \in \calC},\; \calW \defeq \Brace{-\vd_\ell \circ \vk^\star_\ell}_{\ell \in \calL},\] then $\vp_{[T]}$ is an $(\eps_1 + \eps_2)$-omnipredictor for $(\vx_{[T]}, \vy_{[T]}, \calL, \calC)$. 

In the statistical setting, if $\vp: \R^d \to \Delta^k$ satisfies $\eps_1$-$(\calD, \calF)$-multiaccuracy and $\eps_2$-$(\calD, \calW)$-calibration, 
then $\vp$ is an $(\eps_1 + \eps_2)$-omnipredictor for $(\calD, \calL, \calC)$. 
\end{restatable}

As the proof of Proposition~\ref{prop:omni} only slightly modifies the binary case, we defer it to Appendix~\ref{app:defer}.

\section{Simultaneous Blackwell Approachability}\label{sec:meta}

In this section, we develop our main framework, which solves a simultaneous variant of the classical Blackwell approachability problem \cite{Blackwell56} with $m \ge 1$ convex sets. We review the standard $m = 1$ variant in Section~\ref{ssec:blackwell}. We then give our algorithm in Section~\ref{ssec:framework}, which builds upon \cite{AbernethyBH11} to reduce simultaneous Blackwell approachability problems to the implementation of a certain oracle. 

\subsection{Blackwell approachability}\label{ssec:blackwell}

Let $\calA \subseteq \R^a$ and $\calB \subseteq \R^b$ be compact, convex sets. The classical minimax theorem of von Neumann implies that if $f(\va, \vb)$ is a bilinear function of $\va \in \calA$ and $\vb \in \calB$, then 
\begin{equation}\label{eq:vnm}\min_{\va \in \calA} \max_{\vb \in \calB} f(\va, \vb) = \max_{\vb \in \calB} \min_{\va \in \calA} f(\va, \vb).\end{equation}
Seminal work by \cite{Blackwell56} considered the following generalization to vector-valued functions.

\begin{definition}[Satisfiability]\label{def:satisfiability}
Let $\calA \subseteq \R^a, \calB \subseteq \R^b$ be compact and convex, let $\calV \subseteq \R^d$ be convex, and let $\vv: \calA \times \calB \to \R^d$ be a bilinear, vector-valued function.
\begin{itemize}
    \item We call $\calV$ \emph{satisfiable} if there exists $\va \in \calA$ such that $\vv(\va, \vb) \in \vset$ for all $\vb \in \calB$.
    \item We call $\calV$ \emph{response-satisfiable} if for all $\vb \in \calB$ there exists $\va \in \calA$ such that $\vv(\va, \vb) \in \vset$.
    \item We call $\calV$ \emph{halfspace-satisfiable} if all halfspaces containing $\calV$ are satisfiable.
\end{itemize}
\end{definition}

Definition~\ref{def:satisfiability} suggests a natural generalization of \eqref{eq:vnm}: are all response-satisfiable sets also satisfiable? This holds when $d = 1$ by taking $\calV$ to be a sublevel set $(\infty, v]$: in this case, $\calV$ is both satisfiable and response-satisfiable iff $v$ is at least the value of the game \eqref{eq:vnm}. Unfortunately, simple counterexamples (e.g., $\calV = \{(x, x) \mid x \in [0, 1]\}, \vv(a, b) = (a, b)$) preclude this equivalence for $d > 1$. The main result of \cite{Blackwell56} (see also a more modern exposition by \cite{AbernethyBH11}) is that a different equivalence holds.

\begin{proposition}[\cite{Blackwell56}]\label{prop:blackwell_equiv}
In the setting of Definition~\ref{def:satisfiability}, the following three statements are equivalent.
\begin{itemize}
    \item $\calV$ is response-satisfiable.
    \item $\calV$ is halfspace-satisfiable.
    \item $\calV$ is \emph{approachable}, i.e., for any sequence $\{\vb_t\}_{t \in \N}$, there is a choice of $\{\va_t\}_{t \in \N}$ such that $\va_t$ depends only on $\vb_{[t - 1]}$, and such that
    \[\lim_{t \to \infty} \frac 1 t \sum_{s \in [t]} \vv(\va_s, \vb_s) \to \calV.\]
\end{itemize}
\end{proposition}

A quantitative, algorithmic variant of Proposition~\ref{prop:blackwell_equiv} was given by \cite{AbernethyBH11}. To explain its relevance to Section~\ref{ssec:framework}, we specialize our attention to sets $\calV$ induced via a convex set of distinguishers $\calU \subseteq \R^d$, and a scalar $\rho \in \R$. Specifically, we are interested in sets of the form
\begin{equation}\label{eq:sublevel}\calV \defeq \Brace{\vv \in \R^d \mid \sup_{\vu \in \calU} \inprod{\vu}{\vv} \le \rho}.\end{equation}
The function $\sup_{\vu \in \calU} \inprod{\vu}{\vv}$ is called the \emph{support function} of $\calU$ in convex analysis. Intuitively, we can view each member $\vu \in \calU$ as a distinguisher that tests whether a given $\vv \in \R^d$ satisfies $\inprod{\vu}{\vv} \le \rho$. If $\vv$ passes all tests given by $\calU$, then we can certify $\vv \in \calV$ as in \eqref{eq:sublevel}.

Focusing on \eqref{eq:sublevel}, i.e., sublevel sets of support functions, may seem restrictive. We first mention that this specialization captures an important family of sets.

\begin{lemma}\label{lem:polar}
If $\calV \subseteq \R^d$ is closed and convex, $\0_d \in \calV$, and $\rho > 0$, there exists $\calU$ so that \eqref{eq:sublevel} holds.
\end{lemma}
\begin{proof}
It suffices to take $\calU = \rho \calV^\circ$, where $\calV^\circ \defeq \{\vu \in \R^d \mid \sup_{\vv \in \calV} \inprod{\vu}{\vv} \le 1\}$ is the polar set. That $\calV^\circ$ satisfies \eqref{eq:sublevel} with $\rho = 1$ is by the well-known fact $\calV^{\circ\circ} = \calV$ (see e.g., Theorem 14.5, \cite{Rockafellar70book}).
\end{proof}

In fact, specializing to \eqref{eq:sublevel} is the first step in the reduction of \cite{AbernethyBH11} (see their Proposition 2). The key observation of \cite{AbernethyBH11} is that if $\calV$ is a \emph{convex cone} (a set closed under nonnegative linear combinations), then it satisfies \eqref{eq:sublevel} for $\rho = 0$ and $\calU$ taken to be the \emph{dual cone} to $\calV$. Moreover, any convex set $\calV \subseteq \R^d$ can be lifted to a convex cone in $\R^{d + 1}$, by defining the set
\[\clift(\calV) \defeq \Brace{(\vv, c) \in \R^d \times \R_{> 0} \mid \frac{\vv}{c} \in \calV} \cup \Brace{\0_{d + 1}}.\]
Intuitively, the $c = 1$ ``slice'' of $\clift(\calV)$ projects to $\calV$ in the first $d$ dimensions, and the slice at an arbitrary $c \ge 0$ projects to $c\calV$. By converting between the distance of a point in $\R^{d + 1}$ to $\clift(\calV)$, and the distance of its projection in $\R^d$ to $\calV$
(paying an overhead of $\approx \text{diam}(\calV)$), \cite{AbernethyBH11} show that approaching sets \eqref{eq:sublevel} suffices to derive a general algorithmic variant of Proposition~\ref{prop:blackwell_equiv}.

\subsection{Framework}\label{ssec:framework}

We next provide a generalization of Algorithm 2, \cite{AbernethyBH11} that simultaneously approaches a collection of $m$ convex sets of the form \eqref{eq:sublevel}. In fact, under our oracle abstraction (to be introduced in Definition~\ref{def:mloo}), the convex sets we approach need not live in a finite-dimensional space, and can come from an arbitrary Hilbert space. We formalize our problem setting here.

\begin{problem}[Simultaneous Blackwell approachability]\label{prob:sba}
Let $a, b, m \in \N$, let $\rho, \eps \ge 0$, and let $\calA \subseteq \R^a$ and $\calB \subseteq \R^b$ be compact and convex. For all $i \in [m]$, let $\calV^{(i)}, \calU^{(i)} \subseteq \calH^{(i)}$ satisfy
\begin{equation}\label{eq:distinguisher_set}\calV^{(i)} \defeq \Brace{\vv \in \calH^{(i)} \mid \sup_{\vu \in \calU^{(i)}} \inprod{\vu}{\vv} \le \rho},\end{equation}
where $\calH^{(i)}$ is a Hilbert space.
Let $\vv^{(i)}: \calA \times \calB \to \calH^{(i)}$ be a bilinear, vector-valued function for all $i \in [m]$. Our goal is to observe a sequence $\vb_{[T]} \in \calB^T$, and to choose $\va_{[T]} \in \calA^T$ so that $\va_t$ depends on $\vb_{[t - 1]}$ for all $t \in [T]$, and
\begin{equation}\label{eq:simul_reg}\max_{i \in [m]} \sup_{\vu\sps i \in \calU^{(i)}} \inprod{\vu\sps i}{\frac 1 T \sum_{t \in [T]} \vv^{(i)}(\va_t, \vb_t)} \le \rho + \eps.\end{equation}
\end{problem}

When $m = 1$ and $\eps \to 0$, the bound \eqref{eq:simul_reg} implies that $\frac 1 T \sum_{t \in [T]} \vv^{(1)}(\va_t, \vb_t)$ approaches the set $\calV^{(1)}$, because it passes all tests induced by $\calU^{(1)}$. Equivalences between the regret bound \eqref{eq:simul_reg} and distance to $\calV^{(1)}$ are standard, e.g., Lemma 3 of \cite{AbernethyBH11}. In Problem~\ref{prob:sba}, we pose a generalization allowing for $m$ approachability instances, each with their own associated set $\calV^{(i)}$, distinguishers $\calU^{(i)}$, and function $\vv^{(i)}$. The goal \eqref{eq:simul_reg} then asks to approach all $m$  sets simultaneously. 

To achieve simultaneous approachability, we assume existence of the following type of oracle.

\begin{definition}[Mixture linear optimization oracle]\label{def:mloo}
In the setting of Problem~\ref{prob:sba}, we call $\oracle$ an $\eps$-\emph{mixture linear optimization oracle} (MLOO) if on inputs $\vw \in \Delta^m$, and $\{\vu^{(i)}\}_{i \in [m]} \in \prod_{i \in [m]} \calU^{(i)}$, the oracle outputs $\va \in \calA$ satisfying
\begin{equation}\label{eq:mix_halfspace} \sum_{i \in [m]} \vw_i \inprod{\vu^{(i)}}{\vv^{(i)} (\va, \vb)} \le \rho + \eps \text{ for all } \vb \in \calB.\end{equation}
\end{definition}

We briefly interpret Definition~\ref{def:mloo} in the finite-dimensional setting. When the input $\vw$ is a point mass $\ve_i$, the oracle definition is equivalent to finding a distribution $\calP$ such that 
\begin{equation}\label{eq:one_halfspace}\inprod{\vu^{(i)}}{\vv^{(i)}(\va, \vb)} \le \rho + \eps \text{ for all } \vb \in \calB.\end{equation}
Fortunately, we know \eqref{eq:one_halfspace} is achievable (even when $\eps = 0$) whenever $\calV^{(i)}$ is approachable, because for any $\vu^{(i)} \in \calU^{(i)}$, the set $\{\vv \mid \inprod{\vu^{(i)}}{\vv} \le \rho\}$ is a halfspace containing $\calV^{(i)}$. Proposition~\ref{prop:blackwell_equiv} shows $\calV^{(i)}$ is approachable iff it is halfspace-satisfiable, a.k.a.\ there exists $\va$ achieving \eqref{eq:one_halfspace}. 

We next observe that in general, achievability of Definition~\ref{def:mloo} (and our goal \eqref{eq:simul_reg}) are strictly stronger requirements than each $\calV^{(i)}$ being individually approachable.

\begin{lemma}\label{lem:multi_harder}
There exists a simultaneous Blackwell approachability instance (Problem~\ref{prob:sba}) where each individual set $\calV^{(i)}$ is approachable, but a MLOO does not exist, and the simultaneous approachability bound \eqref{eq:simul_reg} is impossible to achieve, for sufficiently small $\eps > 0$.
\end{lemma}
\begin{proof}
We take $m = 2$, $\rho = 0$, $\calU^{(1)} = \calU^{(2)} = \{1\}$ so $\calV^{(1)} = \calV^{(2)} = \{c \in \R \mid c \le 0\}$, $\calA \defeq \Delta^2$, and
\[\vv^{(1)}(\va, \vb) \defeq \va_1,\; \vv^{(2)}(\va, \vb) \defeq \va_2.\]
Clearly these are bilinear functions, and in fact independent of $\vb \in \calB$. Moreover, $\calV^{(1)}$ and $\calV^{(2)}$ are both approachable, the former by repeatedly playing $\va_t \gets \ve_2$, and the latter by repeatedly playing $\vb_t \gets \ve_1$. However, the mixture oracle fails for any $\eps < \half$ by taking $\vw \defeq \half(\ve_1 + \ve_2)$, because for any $\va \in \calA$ (and any $\vb \in \calB$), \eqref{eq:mix_halfspace} would yield the false statement
$\half = \sum_{i \in [2]} \half \cdot 1 \cdot \va_i \le \eps$.

For this same instance, regarding the simultaneous approachability bound \eqref{eq:simul_reg}, no matter what choices of $\{\va_t, \vb_t\}_{t \in [T]}$ were played, the scalars $\frac 1 T \sum_{t \in [T]}\vv^{(1)}(\va_t, \vb_t)$ and $\frac 1 T \sum_{t \in [T]} \vv^{(2)}(\va_t, \vb_t)$ must sum to $1$. Thus, one of the inequalities \eqref{eq:simul_reg}, for $i \in [2]$, must be violated for $\eps < \half$.
\end{proof}

To ease applications, we give a unified recipe for constructing MLOOs in the case of supervised multiclass prediction tasks in Section~\ref{ssec:mixture_oracle}, a setting where we show MLOOs always exist.

We are now ready to present the main result of this section. Our result reduces solving Problem~\ref{prob:sba} to the implementation of an MLOO (Definition~\ref{def:mloo}) and online learners for each $\calU^{(i)}$.

\begin{theorem}[Simultaneous Blackwell approachability]\label{thm:sba}
In the setting of Problem~\ref{prob:sba}, assume we have access to $\oracle$, an $\eps$-MLOO. Further, for all $i \in [m]$ and $T \in \N$, assume there is an online learner $\alg^{(i)}$ that takes inputs $(\va_{[T]}, \vb_{[T]}) \in \calA^T \times \calB^T$, and outputs $\vu_{[T]}^{(i)}$ such that $\vu_t^{(i)} \in \calU\sps i$ depends only on $\va_{[t - 1]}$ and $\vb_{[t - 1]}$, and
    \begin{equation}\label{eq:one_reg_bound}\sup_{\vu^{(i)} \in \calU^{(i)}} \sum_{t \in [T]} \inprod{\vv^{(i)}(\va_t, \vb_t)}{\vu^{(i)} - \vu^{(i)}_t} \le \reg^{(i)}(T),\end{equation}

    for some $\reg^{(i)}: \N \to \R_{\ge 0}$. Finally, assume 
    \begin{equation}\label{eq:width_bound}\Abs{\inprod{\vv^{(i)}(\va, \vb)}{\vu^{(i)}}} \le L\end{equation}
    for all $i \in [m]$, $(\va, \vb) \in \calA \times \calB$, and $\vu^{(i)} \in \calU^{(i)}$. Then, for any $\vb_{[T]} \in \calB^T$, Algorithm~\ref{alg:sba} produces $\va_{[T]} \in \calA^T$ such that $\va_t$ depends only on $\vb_{[t - 1]}$, and 
\[\sup_{\vu^{(i)} \in \calU^{(i)}} \inprod{\vu^{(i)}}{\frac 1 T \sum_{t \in [T]} \vv^{(i)}(\va_t, \vb_t)} \le \rho + \eps + \frac{\reg^{(i)}(T) + L \sqrt{2T\log(m)}}{T}\quad \text{ for all } i \in [m].\]
\end{theorem}
\begin{proof}
The algorithm is presented in Algorithm~\ref{alg:sba}, and we follow the notation therein throughout. By observation, the algorithm computes $\va_t$ on Line~\ref{line:oracle} before observing $\vb_t$. We first show
\begin{equation}\label{eq:middle_layer}
\max_{i \in [m]} \frac 1 T \sum_{t \in [T]} \inprod{\vu_t^{(i)}}{\vv_t^{(i)}} = \sup_{\vw \in \Delta^m} \frac 1 T \sum_{t \in [T]} \inprod{\vw}{\vg_t} \le \rho + \eps + \frac{L\sqrt{2T\log(m)}}{T}.
\end{equation}
Indeed, this follows from
\begin{equation}\label{eq:mwu_det}
\begin{aligned}
\sup_{\vw \in \Delta^m} \frac 1 T \sum_{t \in [T]} \inprod{\vw}{\vg_t} &= \frac 1 T \sum_{t \in [T]} \inprod{\vw_t}{\vg_t} + \sup_{\vw \in \Delta^m} \frac 1 T \sum_{t \in [T]} \inprod{\vw - \vw_t}{\vg_t} \\
&\le \rho + \eps + \frac{L\sqrt{2T\log(m)}}{T}.
\end{aligned}
\end{equation}
Here, we bounded the first term in the right-hand side of the first line by using the oracle guarantee \eqref{eq:mix_halfspace}, which holds for any $\vb_t$ used to define each $\vg_t$. Moreover, to bound the second term, we observe that Lines~\ref{line:mwu_start} and~\ref{line:mwu_end} are implementing multiplicative weight updates, i.e., Lemma~\ref{lem:mirror} with $r(\vw) = \sum_{i \in [m]} \vw_i \log \vw_i$, using feedback vectors $\vg_t$ satisfying $\norm{\vg_t}_\infty \le L$ by assumption \eqref{eq:width_bound}. 

Next we complete the proof: given \eqref{eq:middle_layer} and the regret bound \eqref{eq:one_reg_bound}, for all $i \in [m]$,
\begin{align*}
\sup_{\vu^{(i)} \in \calU^{(i)}} \inprod{\vu^{(i)}}{\frac 1 T \sum_{t \in [T]} \vv_t^{(i)}} &\le \frac 1 T \sum_{t \in [T]} \inprod{\vu_t^{(i)}}{\vv_t^{(i)}} + \sup_{\vu^{(i)} \in \calU^{(i)}} \frac 1 T \sum_{t \in [T]}\inprod{\vu^{(i)} - \vu_t^{(i)}}{\vv_t^{(i)}} \\
&\le \rho + \eps + \frac{\reg^{(i)}(T) + L\sqrt{2T\log(m)}}{T}.
\end{align*}
\end{proof}

\begin{algorithm2e}\label{alg:sba}
\DontPrintSemicolon
\caption{$\sa(\vb_{[T]}, \{\alg^{(i)}\}_{i \in [m]}, \oracle)$}
\textbf{Input:} Online sequence $\vb_{[T]} \in \calB^T$, online learners $\{\alg^{(i)}\}_{i \in [m]}$ satisfying \eqref{eq:one_reg_bound}, $\eps$-MLOO $\oracle$ (following notation in Problem~\ref{prob:sba}, Definition~\ref{def:mloo})\;
\textbf{Output:}  $\va_{[T]} \in \calA^T$ such that each $\va_t$ is output before observing $\vb_t$\;
$\vw_1 \gets \frac 1 m \1_m$\;
$\vu^{(i)}_1 \gets \alg^{(i)}(\{\})$ for all $i \in [m]$\; \tcp*{Initialize each $\vu^{(i)}_1$ as $\alg^{(i)}$ does before observing any examples.}
$\va_1 \gets \oracle(\vw_1, \{\vu_1^{(i)}\}_{i \in [m]})$\;
$\eta \gets \frac 1 L \cdot \sqrt{2\log(m)} \cdot T^{-1/2}$\;
\For{$2 \le t \le T$}
{
$\vv^{(i)}_{t - 1} \gets \vv^{(i)}(\va_{t - 1}, \vb_{t - 1})$ for all $i \in [m]$\;
$\vg_{t - 1} \gets $ vector in $\R^m$ such that $[\vg_{t - 1}]_i = \langle \vu_{t - 1}^{(i)}, \vv_{t - 1}^{(i)}\rangle$ for all $i \in [m]$\;
$\vu^{(i)}_t \gets \alg^{(i)}(\vv_{[t - 1]}^{(i)})$ for all $i \in [m]$\;
$\vw_t \gets \vw_{t - 1} \circ \exp(\eta \vg_{t - 1})$\; \tcp*{$\circ$ denotes entrywise multiplication and $\exp$ is applied entrywise.}\label{line:mwu_start}
$\vw_t \gets \vw_t \norm{\vw_t}_1^{-1}$\;\label{line:mwu_end}
$\va_t \gets \oracle(\vw_t, \{\vu_t^{(i)}\}_{i \in [m]})$\label{line:oracle}
}
\codeReturn $\va_{[T]}$
\end{algorithm2e}

Theorem~\ref{thm:sba} in fact holds in a much more general setting than captured by simultaneous Blackwell approachability (Problem~\ref{prob:sba}). Indeed, both the statement of Theorem~\ref{thm:sba} and the MLOO definition do not explicitly specify sets $\calV^{(i)}$ distinguished by the corresponding $\calU^{(i)}$ in the sense of \eqref{eq:distinguisher_set}. For our applications to online omniprediction, we require a more general version of \Cref{thm:sba}, where each $\calU^{(i)}$ consists of functions $\vu\sps i:\calX\to \calH\sps i$, taking as input a context $\vx$ from some domain $\calX$. For these uses, we generalize Problem~\ref{prob:sba} and Definition~\ref{def:mloo}, and give an analog of Theorem~\ref{thm:sba}.

\begin{problem}[Contextual Blackwell approachability]\label{prob:csba}
Let $a, b, m \in \N$, let $\eps \ge 0$, let $\calA$ and $\calB$ be subsets of vector spaces, and let $\calX$ be an abstract domain of contexts. For all $i \in [m]$, let $\calU\sps i$ consist of functions $\vu\sps i: \calX \to \calH\sps i$, where $\calH\sps i$ is a Hilbert space. Let $\vv\sps i: \calA \times \calB \to \calH\sps i$ be a bilinear, vector-valued function for all $i \in [m]$. Our goal is to observe sequences $\vb_{[T]} \in \calB^T$ and $\vx_{[T]} \in \calX^T$, and to choose $\va_{[T]}$ so that $\va_t$ depends on $\vb_{[t - 1]}$ and $\vx_{[t]}$ for all $t \in [T]$, and
\[\max_{i \in [m]} \sup_{\vu\sps i \in \calU\sps i} \frac 1 T \sum_{t \in [T]}\inprod{\vu\sps i(\vx_t)}{\vv\sps i(\va_t, \vb_t)} \le \eps.\]
When each $\calU\sps i$ consists only of constant functions, i.e., each $\vu^{(i)} \in \calU^{(i)}$ can only take on one value, we abuse notation and let $\calU^{(i)} \subseteq \calH\sps i$ represent a set of $\vu\sps i \in \calH\sps i$.
\end{problem}

\begin{definition}[Contextual mixture linear optimization oracle]\label{def:cmloo}
In the setting of Problem~\ref{prob:csba}, we call $\oracle$ an $\eps$-\emph{contextual mixture linear optimization oracle} (CMLOO) if on inputs $\vw \in \Delta^m$, $\vx \in \calX$, and $\{\vu^{(i)}\}_{i \in [m]} \in \prod_{i \in [m]} \calU^{(i)}$, the oracle outputs $\va \in \calA$ satisfying
\begin{equation}\label{eq:mix_halfspace_context} \sum_{i \in [m]} \vw_i \inprod{\vu^{(i)}(\vx)}{\vv^{(i)} (\va, \vb)} \le \eps \text{ for all } \vb \in \calB.\end{equation}
\end{definition}

\begin{algorithm2e}\label{alg:csba}
\DontPrintSemicolon
\caption{$\csa(\vb_{[T]}, \vx_{[T]}, \{\alg^{(i)}\}_{i \in [m]}, \oracle)$}
\textbf{Input:} Online sequences $\vb_{[T]} \in \calB^T$, $\vx_{[T]} \in \calX^T$, online learners $\{\alg^{(i)}\}_{i \in [m]}$ satisfying \eqref{eq:one_reg_bound_context}, $\eps$-CMLOO $\oracle$ (following notation in Problem~\ref{prob:csba}, Definition~\ref{def:cmloo})\;
\textbf{Output:}  $\va_{[T]} \in \calA^T$ such that each $\va_t$ is output after observing $\vx_t$ and before observing $\vb_t$\;
$\vw_1 \gets \frac 1 m \1_m$\;
$\vu^{(i)}_1 \gets \alg^{(i)}(\{\})$ for all $i \in [m]$\; 
$\va_1 \gets \oracle(\vw_1, \{\vu_1^{(i)}\}_{i \in [m]})$\;
$\eta \gets \frac 1 L \cdot \sqrt{2\log(m)} \cdot (5T)^{-1/2}$\;
\For{$2 \le t \le T$}
{
$\vv^{(i)}_{t - 1} \gets \vv^{(i)}(\vp_{t - 1}, \vb_{t - 1})$ for all $i \in [m]$\;
$\tvg_{t - 1} \gets $ vector in $\R^m$ such that $[\tvg_{t - 1}]_i = \langle \vu_{t - 1}^{(i)}(\vx_{t - 1}), \vv_{t - 1}^{(i)}\rangle$ for all $i \in [m]$\;
$\vu^{(i)}_t \gets \alg^{(i)}(\vv^{(i)}_{[t - 1]})$ for all $i \in [m]$\;
$\vw_t \gets \vw_{t - 1} \circ \exp(\eta \tvg_{t - 1})$\; \label{line:mwu_start}
$\vw_t \gets \vw_t \norm{\vw_t}_1^{-1}$\;\label{line:mwu_end}
$\va_t \gets \oracle(\vw_t, \vx_t, \{\vu_t^{(i)}\}_{i \in [m]})$\;\label{line:oracle}
$\vp_t \gets $ any random element of $\calA$ such that $\E[\vp_t \mid \va_{[t - 1]}, \vb_{[t-1]}] = \va_t$
}
\codeReturn $\vp_{[T]}$
\end{algorithm2e}

We now state our extension of Theorem~\ref{thm:sba}. 

In our statement, we allow for unbiased estimators of the CMLOO outputs to be played, and give a high-probability guarantee on the error. We also allow for improper learners that satisfy \eqref{eq:one_reg_bound_context}, but output hypotheses from an (appropriately bounded) different set than $\calU\sps i$.

\begin{corollary}\label{cor:context-sba}
In the setting of Problem~\ref{prob:csba}, assume we have access to $\oracle$, an $\eps$-CMLOO. Further, for all $i \in [m]$ and $T \in \N$, assume there is an online learner $\alg\sps i$ that takes inputs $(\va_{[T]}, \vb_{[T]}, \vx_{[T]}) \in \calA^T \times \calB^T \times \calX^T$, and outputs $\vu_{[T]}\sps i \in ((\calU')\sps i)^T$ such that $\vu_t\sps i$ depends only on $\va_{[t - 1]}$, $\vb_{[t - 1]}$, and $\vx_{[t]}$, and
\begin{equation}\label{eq:one_reg_bound_context}\sup_{\vu\sps i \in \calU\sps i} \sum_{t \in [T]} \inprod{\vv\sps i(\va_t, \vb_t)}{\vu\sps i(\vx_t) - \vu_t\sps i(\vx_t)} \le \reg\sps i(T),\end{equation}
for some $\reg\sps i: \N \to \R_{\ge 0}$. Finally, assume     \begin{equation}\label{eq:width_bound_context}\Abs{\inprod{\vv^{(i)}(\va, \vb)}{\vu^{(i)}(\vx)}} \le L\end{equation}
    for all $i \in [m]$, $(\va, \vb, \vx) \in \calA \times \calB \times \calX$, and $\vu^{(i)} \in \calU^{(i)} \cup (\calU')\sps i$. Then, for any $(\vb_{[T]}, \vx_{[T]}) \in \calB^T \times \calX^T$, Algorithm~\ref{alg:csba} produces $\vp_{[T]} \in \calA^T$ such that $\vp_t$ depends only on $\vb_{[t - 1]}$ and $\vx_{[t]}$, and for any $\delta \in (0, 1)$,
\begin{equation}\label{eq:conclusion_context}\sup_{\vu^{(i)} \in \calU^{(i)}} \frac 1 T \sum_{t \in [T]}\inprod{\vu^{(i)}(\vx_t)}{ \vv^{(i)}(\vp_t, \vb_t)} \le \eps + \frac{\reg^{(i)}(T) + 28L \sqrt{T\log(\frac {4m}{\delta})}}{T}\text{ for all } i \in [m],\end{equation}
with probability at least $1 - \delta$ over the randomness of the $\vp_{[T]}$.
\end{corollary}
\begin{proof}
The proof is largely analogous to Theorem~\ref{thm:sba}, substituting the contexts $\vx_{[T]}$ as necessary. The key difference is in \eqref{eq:mwu_det}, which no longer holds deterministically. We instead apply the variant in Lemma~\ref{lem:high_prob_sgd} with the $\tvg_{[T]}$ as defined in Algorithm~\ref{alg:csba}. Note that this is unbiased for $\vg_t$ with entries $\inprods{\vu_t\sps i}{\vv\sps i(\va_t, \vb_t)}$, by linearity of each $\vv\sps i$ in its first argument. Thus, in place of \eqref{eq:mwu_det},
\begin{align*}\sup_{\vw \in \Delta^m} \frac 1 T \sum_{t \in [T]} \inprod{\vw}{\tvg_t} &= \frac 1 T \sum_{t \in [T]} \inprod{\vw_t}{\vg_t} + \frac 1 T \sum_{t \in [T]} \inprod{\vw_t}{\tvg_t - \vg_t} + \sup_{\vw \in \Delta^m} \frac 1 T \sum_{t \in [T]} \inprod{\vw - \vw_t}{\tvg_t}\\
&\le \eps + \frac 1 T \sum_{t \in [T]} \inprod{\vw_t}{\tvg_t - \vg_t} + 20L\sqrt{\frac{\log(\frac{4m}{\delta})}{T}} \le \eps + 28L\sqrt{\frac{\log(\frac{4m}{\delta})}{T}},
\end{align*}
with probability $\ge 1 - \delta$. Here, the first inequality used the CMLOO guarantee to bound $\inprod{\vw_t}{\vg_t} \le \eps$ for all $t \in [T]$, as well as Lemma~\ref{lem:high_prob_sgd} with failure probability $\frac \delta 2$ to bound the last term. The second inequality used the Azuma-Hoeffding inequality with the fact that each $\inprod{\vw_t}{\tvg_t - \vg_t}$ is mean-zero conditioned on the history, and bounded in $[-2L, 2L]$.  We remark that the setting of $\eta$ in Algorithm~\ref{alg:csba} is for Lemma~\ref{lem:high_prob_sgd} to hold; see Lemma 9, \cite{hu2025omnipredicting} for additional discussion. 
\end{proof}

\section{Binary Omniprediction}\label{sec:binary}

As a warmup, we consider the binary omniprediction setting, where we make predictions over $k = 2$ classes. We state some general preliminaries in Section~\ref{sec:binary_prelim}. In Sections~\ref{ssec:binary_online} and~\ref{ssec:binary_stat} we apply our framework from Section~\ref{sec:meta} to reduce online and statistical omniprediction, respectively, to appropriate low-regret learners. We complete our binary omniprediction results by providing these low-regret learners in the linear (Section~\ref{ssec:binary_linear}) and general (Section~\ref{ssec:binary_general}) classification settings. 

For notational simplicity in this section only, we identify the binary simplex $\Delta^2$ with the prediction interval $[0, 1]$ and the boundary $\partial \Delta^2$ with the label set $\{0, 1\}$, so e.g., in place of \eqref{eq:glm_def},
\begin{equation}\label{eq:glm_def_2}\Lglm \defeq \Brace{\ell: \R \times \{0, 1\} \to \R \mid \ell(t, y) = \omega(t) - ty,\; \omega: \R \to \R \text{ convex with } \omega': \R \to [0, 1]}.\end{equation}
We will also use standard script (rather than boldface) to denote any scalar-valued variables. Finally, throughout the section we make the normalization assumptions that
\begin{equation}\label{eq:c_l_bound}
\begin{gathered}
\calX \subseteq \ball_2^d(1),\quad c(\vx) \in [-1, 1] \text{ for all } c \in \calC,\; \vx \in \calX,\\
\ell(t, y) \in [-1, 1] \text{ for all } \ell \in \calL,\; (t, y) \in [-1, 1] \times \{0, 1\}.
\end{gathered}
\end{equation}
All our results generalize to generic bounds on $\calL$ and $\calC$ in a scale-invariant way. For example, the assumption \eqref{eq:c_l_bound} is enforced for $\Lglm$ by requiring that $\omega'(t) \in [0, 1]$ for all $t \in [-1, 1]$.

\subsection{Binary omniprediction preliminaries}\label{sec:binary_prelim}

We recall two useful characterizations of proper losses from prior work.

\begin{lemma}[Lemma 3, \cite{KleinbergLST23}]\label{lem:proper_1}
Let $\ell: \Omega \times \partial \Delta^k \to \R$ be arbitrary. Then defining $\vk^\star_\ell$ as in \eqref{eq:eao}, the function $(\vp, \vy) \to \ell(k_\ell^\star(\vp), \vy)$ is a proper loss.
\end{lemma}

\begin{lemma}[Theorem 8, \cite{KleinbergLST23}]\label{lem:proper_2}
For all $s \in [0, 1]$ and $(p, y) \in [0, 1]^2$, let
\begin{equation}\label{eq:thresh_loss}\ell_s(p, y) \defeq -\Abs{p - s} + (p-y)\sign(p-s).\end{equation}
Then $\ell_s$ is proper for all $s \in [0, 1]$, and for every bounded proper loss $\ell: [0, 1] \times \{0, 1\} \to [-1, 1]$,
\[\ell(p,y) = a_\ell y + b_\ell + \int w_\ell(v) \ell_v(p,y) \dd v,\]
for some nonnegative weights $\{w_\ell(v)\}_{v \in [0, 1]}$ satisfying $\int_0^1 w_\ell(v) \dd v \le 2$, and a constant $|a_\ell| \le 2$.
\end{lemma}
\begin{proof}
All parts of the claim are explicit in Theorem 8, \cite{KleinbergLST23}, except for the bound on $|a_\ell|$, so we reproduce part of the proof here to make this clear. Overloading notation, any (bivariate) proper loss $\ell: [0, 1] \times \{0, 1\} \to [-1, 1]$ can be written as 
\[\ell(p,y) = -\uni(\ell)(p)+(p-y)\uni(\ell)'(p)\]
for a univariate convex function $\uni(\ell): [0, 1] \to \R$ with $|\uni(\ell)'| \le 2$ (this is implied by Lemma~\ref{lem:glm_proper} with $\uni(\ell) \gets \psi$; see also Lemma 1 and Corollary 2, \cite{KleinbergLST23}). Since piecewise linear functions are dense for this family of $\ell$ under the $\sup$ norm, we only need to consider $\uni(\ell)(p)$ that is piecewise linear with respect to $p$. If $\psi \defeq \uni(\ell)$ has $k$ breakpoints, $s_{1},\ldots, s_{k}$, we can write $\psi(p)$ as
    \[\psi(p)=\psi(0)+p\psi'(0)+\sum_{i\in[k]}\Par{\psi'_{+}(s_i)-\psi_{-}'(s_i)}\cdot\max(p-s_i, 0),\]
    where $\psi_{+}'(s) = \lim_{p\rightarrow s^{+}}\ell'(p)$ and $\ell_{-}'(s) = \lim_{p\rightarrow s^{-}}\ell'(p)$. The claim then follows since $\max(x, 0) = \thalf\Par{\Abs{x}+x}$ and $\psi_{s_i}(p) = \Abs{p-s_i}$ is the univariate form of $\ell_{s_i}$ in \eqref{eq:thresh_loss}, so equivalently,
    \[
        \psi(p)=\Par{\psi(0)-\frac{1}{2}\sum_{i\in[k]}\lambda_i s_i}+p\Par{\psi'(0)+\frac{1}{2}\sum_{i\in[k]}\lambda_i}+\frac{1}{2}\sum_{i\in[k]}\lambda_i\cdot\psi_{s_i}(p),
    \]
    with $\lambda_i\defeq \psi'_{+}(s_i)-\psi_{-}'(s_i) \ge 0$. Note that since $|\psi'(0) + \sum_{i\in[k]}\lambda_i| = \Abs{\psi'(s_k)}\leq 2$ and $|\psi'(0)| \le 2$, the coefficient of $p$ has absolute value bounded by $2$ as claimed.
\end{proof}

By combining Lemmas~\ref{lem:proper_1} and~\ref{lem:proper_2}, \cite{OkoroaforKK25} showed that obtaining calibration against the family of weights $\{-d_\ell \circ k^\star_\ell\}_{\ell \in \calL}$, as required by Proposition~\ref{prop:omni}, can be reduced to calibration against an appropriate basis of weights induced by the specific proper losses in \eqref{eq:thresh_loss}. 

\begin{lemma}\label{lem:thresh_suffice}
Let $\calL$ be a family of losses over $\Omega \times \{0, 1\}$ such that $\ell(\omega, y) \in [-1, 1]$ for all $\ell \in \calL$, $\omega \in \Omega$, $y \in \{0, 1\}$. In the statistical setting, if $p: \R^d \to [0, 1]$ satisfies $\eps$-$(\calD, \Wthresh)$-calibration for
\begin{equation}\label{eq:wthresh_def}\Wthresh \defeq \Brace{w(p) = \sign(p-s)}_{s \in [0, 1]},\end{equation}
it also satisfies $4\eps$-$(\calD, \{-\dd_\ell \circ k^\star_\ell\}_{\ell \in \calL})$-calibration. In the online setting, if $p_{[T]} \in [0, 1]^T$ satisfies $\eps$-$(y_{[T]}, \Wthresh )$-calibration, it also satisfies $4\eps$-$(y_{[T]}, \{-\dd_\ell \circ k^\star_\ell\}_{\ell \in \calL})$-calibration.
\end{lemma}
\begin{proof}
The proof follows Corollary 3.4 in \cite{OkoroaforKK25}, but includes the linear term $a_\ell y$ (this term was not explicitly discussed in \cite{OkoroaforKK25}, which we account for here). Specifically, the proof of Theorem 3.1 in \cite{OkoroaforKK25} shows that Lemma~\ref{lem:proper_1} implies it suffices to prove the lemma statement for proper $\ell$, and drop composition with $k^\star_\ell$. Next, for any proper $\ell$,
\[d_\ell(p) = a_\ell + \int w_\ell(v) d_{\ell_v}(p) \dd v,\]
for nonnegative weights $\{w_\ell(v)\}_{v \in [0, 1]}$ satisfying $\int_0^1 w_\ell(v) \dd v \le 2$, and $|a_\ell| \le 2$, by using Lemma~\ref{lem:proper_2} and linearity of $d_\ell$. Equivalently, because $d_{\ell_0} = \sign(0 - \cdot) = -1$ and $d_{\ell_1} = \sign(1 - \cdot) = 1$ are constant functions, we can directly include $a_\ell$ in an appropriate weight in the integral, so that 
\[d_\ell(p) = \int u_\ell(v) d_{\ell_v}(p) \dd v,\]
for nonnegative $\{u_\ell(v)\}_{v \in [0, 1]}$ satisfying $\int_0^1 u_\ell(v) \dd v \le 4$.
The conclusion follows because $\eps$-calibration against $\Wthresh$ implies $4\eps$-calibration against any such $d_\ell$ defined above, by integrating.
\end{proof}

\subsection{Online binary omniprediction}\label{ssec:binary_online}
We first consider the online binary omniprediction setting, i.e., Definition~\ref{def:omni} with $k=2$, for a family of losses $\calL$ and a family of comparators $\calC$. Throughout, let $\calX \subseteq \R^d$ be the support of our features. 

Our starting point is an observation from \cite{OkoroaforKK25} that the (online variants of) Definitions~\ref{def:ma} and~\ref{def:cal} can naturally be framed in the context of Problem~\ref{prob:csba}. 
For a fixed parameter $\eps \in (0, 1)$,
we define
\begin{equation}\label{eq:abdef_binary_online}\calA \defeq \Delta^{\calN},\quad \calB \defeq [0, 1],\end{equation}
where $\calN \defeq \{i \eps\}_{i \in [\lfloor\frac 1 \eps \rfloor]} \cup \{0, 1\}$ is an $\eps$-net for $[0, 1]$, and is viewed as a set of representative thresholds. 
For simplicity of notation, we use $p \sim \va$ to mean that $p \in \calN$ is sampled according to $\va$. Note that the sequence $b_{[T]} \in \calB^T$ will eventually correspond to our (binary) online label space.

We next define the sets and payoff vectors in Problem~\ref{prob:csba}, where $\sign(0) = 1$ by convention:
\begin{equation}\label{eq:uvdef_binary_online}
\begin{gathered}\calU^{(1)} \defeq \Delta^{\calN},\quad \vv^{(1)}(\va, b) \defeq \Brace{\E_{p \sim \va}\Brack{\Par{p - b}\sign(p-s)}}_{s \in \calN}, \\
\calU^{(2)} \defeq \Brace{\dd_\ell \circ c}_{\ell \in \calL, c \in \calC},\quad v^{(2)}(\va, b) \defeq \E_{p \sim \va}\Brack{p-b}.
\end{gathered}
\end{equation}

We remark that our definition of $\calU\sps 2$ is exactly the set $\calF$ in Proposition~\ref{prop:omni}.

We equip both $\calH^{(1)} = \R^{\calN}$ and $\calH^{(2)} = \R$ with the standard Euclidean inner product. Note that $\vv^{(1)}$ is a function that takes $(\va, b) \in \calA \times \calB$ to a vector in $\R^{\calN}$, whose coordinates are indexed by $\calN$. Each coordinate of $\calU\sps 1$ is used to ensure calibration against an element of $\Wthresh$. In a slight abuse of notation (as remarked on in Problem~\ref{prob:csba}), $\calU\sps 1$ consists of elements of $\calH\sps 1$, whereas $\calU\sps 2$ consists of functions taking contexts from our domain $\calX$ to $\calH\sps 2$.

To use Corollary~\ref{cor:context-sba}, we first instantiate the learner $\alg\sps 1$ as specified in \eqref{eq:one_reg_bound_context}.

\begin{lemma}\label{lem:online_regret_bi_cal}
Following definitions \eqref{eq:abdef_binary_online}, \eqref{eq:uvdef_binary_online}, there exists $\alg\sps 1$ such that for any $(\va_{[T]}, b_{[T]}) \in \calA^T \times \calB^T$, $\alg\sps 1$ outputs $\vu_{[T]}\sps 1 \in (\calU\sps 1)^T$ such that $\vu_t\sps 1$ depends only on $\va_{[t - 1]}, b_{[t - 1]}$, and \eqref{eq:one_reg_bound_context} holds with
\[\reg\sps 1(T) \defeq \sqrt{2T\log\Par{\frac 1 \eps + 2}}.\]

\end{lemma}
\begin{proof}

The algorithm is multiplicative weights. More precisely, for all $(\va_t, b_t) \in \calA \times \calB$,
\[\norms{\vv\sps 1(\va_t, b_t)}_\infty \le 1.\]
Therefore, Lemma~\ref{lem:mirror} applies with $L = 1$ and $\norm{\cdot} = \norm{\cdot}_1$. We choose $r(\vu) \defeq \sum_{s \in \calN} \vu_s \log \vu_s$, at which point the standard bound $\Theta \le \log(|\calN|)$ in Lemma~\ref{lem:mirror} yields the conclusion.
\end{proof}

Next, we recall a construction of a CMLOO (Definition~\ref{def:cmloo}) given by \cite{OkoroaforKK25}.

\begin{lemma}[Lemma 3.11,~\cite{OkoroaforKK25}]\label{lem:oracle_binary}
Following definitions \eqref{eq:abdef_binary_online}, \eqref{eq:uvdef_binary_online}, and assuming \eqref{eq:c_l_bound}, there exists $\oracle$, an $\eps$-CMLOO.
\end{lemma}
\begin{proof}
We state the algorithm in Algorithm~\ref{alg:oracle_binary}. For notational simplicity, we fix inputs $\vw = (q, r) \in \Delta^2$, $\vx \in \calX$, $\vu \defeq \vu\sps 1 \in \Delta^{\Abs{\calN}+2}$, and $\dd_\ell \circ c \defeq \vu\sps 2 \in \calU\sps 2$ to the CMLOO. Also, let
\begin{align*}
f(\va, b) \defeq q \Par{\sum_{s \in \calN} \vu_s \E_{p \sim \va}\Brack{(p - b)\sign(p-s)} } + r \dd_\ell(c(\vx)) \E_{p \sim \va}[p-b],
\end{align*}
so that following Definition~\ref{def:cmloo}, we wish to find $\va \in \calA$ such that $f(\va, b) \le \eps$ for every $b \in \calB$. It is helpful to define the ``pure strategy'' specialization of \eqref{eq:cmloo_helper}, i.e., where $\va$ is a point mass on $s \in \calN$:
\begin{equation}\label{eq:cmloo_helper}h(p) \defeq q\Par{\sum_{s \in \calN}\vu_s\sign(p-s)} + r\dd_\ell(c(\vx)).\end{equation}
In particular, we have $f(\ve_p, b) = h(p)(p-b)$ in this special case. Observe that under \eqref{eq:c_l_bound}, we have for any $p \in \calN$ that $|h(p)| \le 1$, because $|c(\vx)| \le 1$, $|\sign(p-s)| \le 1$ for all $s \in \calN$, and $q + r = 1$.

\textbf{Case 1: $h(0) \ge 0$.} In this case, Algorithm~\ref{alg:oracle_binary} outputs $\va = \ve_0$, which satisfies
\[f(\va, b) = h(0)(0-b) \le 0, \text{ for all } b \in \calB.\]

\textbf{Case 2: $h(1) \le 0$.} In this case, we similarly have for $\va = \ve_1$,
\[f(\va, b) = h(1)(1-b) \le 0, \text{ for all }b \in \calB. \]

\textbf{Case 3: $h(0) < 0$ and $h(1) > 0$.} In this case, there are adjacent $(p, p') \in \calN \times \calN$ with $p \le p' \le p + \eps$, $h(p) \le 0$, and $h(p') \ge 0$, and Algorithm~\ref{alg:oracle_binary} outputs $\va = \frac{|h(p')|}{|h(p)| + |h(p')|}\ve_p + \frac{|h(p)|}{|h(p)|+|h(p')|} \ve_{p'}$. Then,
\begin{align*}
f(\va, b) &= \frac{|h(p')|}{|h(p)| + |h(p')|} \cdot h(p)(p-b) + \frac{|h(p)|}{|h(p)| + |h(p')|} \cdot h(p')(p'-b) \\
&= \frac{|h(p')|}{|h(p)| + |h(p')|} \cdot h(p)(p-b) + \frac{|h(p)|}{|h(p)| + |h(p')|} \cdot h(p')(p-b) \\
&+ \frac{|h(p)|}{|h(p)| + |h(p')|} \cdot h(p') (p' - p) \\
&= \frac{|h(p)|}{|h(p)| + |h(p')|} \cdot h(p')(p'-p) \le \eps,
\end{align*}
for all $b \in \calB$, where the last line used $|p - p'| \le \eps$ and $|\frac{h(p)}{|h(p)|+|h(p')|} \cdot h(p')| \le 1$.
\end{proof}
\SetKw{KwRet}{return}
\begin{algorithm2e}[h]\label{alg:oracle_binary}
\DontPrintSemicolon
\caption{CMLOO for binary omniprediction}
\textbf{Input:} $\vw = (q, r) \in \Delta^2$, $\vx \in \calX$, $\vu \in \Delta^{\calN}$, $c \in \calC$, $\ell \in \calL$\;
\lIf{$h(0) \ge 0$}{%
  \KwRet $\ve_0$\tcp*[f]{Following definition \eqref{eq:cmloo_helper}.}%
}
\lElseIf{$h(1) \le 0$}{\KwRet $\ve_1$}
\uElse{
$(p, p') \gets $ elements in $\calN \times \calN$ such that $p \le p' \le p + \eps$, $h(p) \le 0$, $h(p') \ge 0$\;
\Return{$\frac{|h(p')|}{|h(p)|+|h(p')|}\ve_p + \frac{|h(p)|}{|h(p)|+|h(p')|}\ve_{p'}$}
}
\end{algorithm2e}    

We conclude the section by showing how to apply Lemmas~\ref{lem:online_regret_bi_cal} and~\ref{lem:oracle_binary} within the context of Proposition~\ref{prop:omni} to give our result on online binary omniprediction.

\begin{corollary}[Online binary omniprediction]\label{cor:online_bi}
Let $\calL$ be a family of loss functions and $\calC$, $\calC'$ be families of comparators satisfying \eqref{eq:c_l_bound}. 
Assume there exists an online learner $\alg\sps 2$ that takes inputs $(v_{[T]}, \vx_{[T]}) \in [-1, 1]^T \times \calX^T$, and outputs $\ell_{[T]} \in \calL^T$, $c_{[T]} \in (\calC')^T$,\footnote{We include this additional flexibility for our applications in Section~\ref{ssec:binary_general}, which may return improper hypotheses.} such that $(\ell_t, c_t)$ depends only on $v_{[t - 1]}$, $\vx_{[t - 1]}$, and
\begin{equation}\label{eq:ma_regret}
\sup_{(\ell, c) \in \calL \times \calC} \sum_{t \in [T]} v_t(\dd_\ell(c(\vx_t)) - \dd_{\ell_t}(c_t(\vx_t))) \le \reg(T),
\end{equation}
for $\reg: \N \to \R_{\ge 0}$ such that all $T \ge T_{\calL, \calC}$ satisfy $\frac{\reg(T)}{T} \le \eps$.
Then if $T = \Omega(\frac{1}{\eps^2} \log(\frac 1{\delta\eps})) + T_{\calL, \calC}$, we can produce $p_{[T]} \in [0, 1]^T$, a $15\eps$-omnipredictor for $(\vx_{[T]}, y_{[T]}, \calL, \calC)$, with probability $\ge 1 - \delta$.
\end{corollary}
\begin{proof}
Following notation in Proposition~\ref{prop:omni}, it suffices to show how to produce $p_{[T]} \in [0, 1]^T$ satisfying $3\eps$-$(\vx_{[T]}, y_{[T]}, \calF)$-multiaccuracy and $12\eps$-$(y_{[T]}, \calW)$-calibration, where we recall
\[\calF \defeq \Brace{\dd_\ell \circ c}_{\ell \in \calL, \vc \in \calC},\; \calW \defeq \Brace{-\dd_\ell \circ k^\star_\ell}_{\ell \in \calL}.\]
We achieve this by using Corollary~\ref{cor:context-sba}. From our definitions of $\calU\sps 1$, $\vv\sps 1$, $\calU\sps 2$, and $\vv\sps 2$ in \eqref{eq:uvdef_binary_online}, it is clear that we may take $L = 1$ in \eqref{eq:width_bound_context}. Now if we can satisfy the requirements \eqref{eq:one_reg_bound_context} of Corollary~\ref{cor:context-sba}, playing any $p_t \sim \va_t$ as $\va_t$ is produced by the CMLOO in Lemma~\ref{lem:oracle_binary} in each iteration, gives
\begin{equation}\label{eq:calma_binary_online}\sup_{\vu\sps i \in \calU\sps i} \frac 1 T \sum_{t \in [T]} \inprod{\vu\sps i(\vx_t)}{\vv\sps i(\ve_{p_t}, b_t)} \le \eps + \frac{\reg^{(i)}(T) + 28 \sqrt{T\log(\frac {8}{\delta})}}{T} \end{equation}
for $i \in [2]$ with probability $\ge 1 - \delta$. Condition on this event henceforth.

When $i = 1$, the guarantee in \eqref{eq:calma_binary_online} exactly corresponds to $(y_{[T]},\Wthresh)$-calibration, when comparing with Definition~\ref{def:cal} and Lemma~\ref{lem:thresh_suffice}.\footnote{Our definition of $\calU\sps 1$ formally only yields $(y_{[T]},\Wthresh)$-calibration for the thresholds $s \in \calN$. However, because we only play predictions in $\calN$, all weights induced by thresholds outside $\calN$ agree with that of some threshold in $\calN$.} Similarly, when $i = 2$, \eqref{eq:calma_binary_online} is exactly $(\vx_{[T]}, y_{[T]}, \calF)$-multiaccuracy. Thus, it is enough to bound the right-hand side of \eqref{eq:calma_binary_online} by $3\eps$ in both cases, which also results in $12\eps$-$(y_{[T]}, \calW)$-calibration via Lemma~\ref{lem:thresh_suffice}.
To do this we use the online learner from Lemma~\ref{lem:online_regret_bi_cal} as $\alg\sps 1$, and the learner with guarantee \eqref{eq:ma_regret} as $\alg\sps 2$, and take $T$ as specified.
\end{proof}

We postpone discussion of the construction of online learners $\alg\sps 2$ meeting the requirement \eqref{eq:ma_regret}, for both specific and general pairs of $\calL \times \calC$, to Sections~\ref{ssec:binary_linear} and~\ref{ssec:binary_general}.

\subsection{Statistical binary omniprediction}\label{ssec:binary_stat}
We next consider statistical omniprediction, again with $k = 2$. However, we require a different formulation of Problem~\ref{prob:csba}. 
In this section, we let $\calH\sps 1$ and $\calH\sps 2$ be the Hilbert spaces of (norm) square-integrable functions under $\calD$, taking $\calX \times \{0, 1\} \to \R^{\calN}$ and $\calX \times \{0, 1\} \to \R$ respectively. The inner products of $\vu, \vv \in \calH\sps 1$ and $u, v \in \calH\sps 2$ are the corresponding $L^2(\calD)$ inner products:
\[\inprod{\vu}{\vv} \defeq \E_{(\vx, y) \sim \calD}\Brack{\inprod{\vu(\vx, y)}{\vv(\vx,y)}},\quad \inprod{u}{v} \defeq \E_{(\vx, y)}\Brack{u(\vx, y)v(\vx, y)}.\]
Note that in this instance of Problem~\ref{prob:csba}, there is no additional context $\vx \in \calX$, as it is implicitly specified through our inner product definitions. 

Next, define $\calA$ to be the set of functions taking each $\vx \in \calX \to \va(\vx) \in \Delta^{\calN}$, i.e.,
\begin{equation}\label{eq:adef_binary_stat}\calA \defeq \Brace{\va: \calX \to \Delta^{\calN}}.\end{equation}
As before, $\calN$ is an $\eps$-net for $[0, 1]$ that includes $\{0, 1\}$, and we write $p \sim \va(\vx)$ to mean $p \in \calN$ is sampled as specified by $\va(\vx) \in \Delta^{\calN}$. Also, our payoff vectors $\vv\sps 1$ and $\vv\sps 2$ will be independent of $b \in \calB$, so we simply let $\calB = \emptyset$ and drop the input $b$ from $\vv\sps 1$, $\vv\sps 2$.

Finally, we let
\begin{equation}\label{eq:uvdef_binary_stat}
\begin{aligned}
\calU\sps 1 &\defeq \Delta^{\calN},\quad \vv\sps 1(\va)(\vx, y) \defeq \Brace{\E_{p \sim \va(\vx)}\Brack{(p - y) \sign(p-s)} }_{s \in \calN}, \\
\calU\sps 2 &\defeq \Brace{\dd_\ell \circ c}_{\ell \in \calL, c \in \calC}, \quad v\sps 2(\va)(\vx, y) \defeq \E_{p \sim \va(\vx)}[p - y].
\end{aligned}
\end{equation}

As before, $\calU\sps 1$ is interpreted as the family of constant functions over $\calX \times \{0, 1\}$ with range $\Delta^{\calN}$, and $\dd_\ell \circ c \in \calU\sps 2$ acts on $(\vx, y) \in \calX \times \{0, 1\}$ by discarding $y$ and outputting $d_\ell(c(\vx))$. We also specify the functions $\vv\sps 1(\va) \in \calH\sps 1$ and $\vv\sps 2(\va) \in \calH\sps 2$ by their actions on an element $(\vx, y) \in \calX \times \{0, 1\}$.

With this setup in hand, we now extend Lemmas~\ref{lem:online_regret_bi_cal} and~\ref{lem:oracle_binary} to the statistical setting.

\begin{lemma}\label{lem:offline_regret_bi_cal}
Let $\delta \in (0, 1)$. Following definitions \eqref{eq:adef_binary_stat}, \eqref{eq:uvdef_binary_stat}, there exists $\alg\sps 1$ such that for any $\va_{[T]} \in \calA^T$, $\alg\sps 1$ outputs $\vu_{[T]}\sps 1 \in (\calU\sps 1)^T$ such that $\vu_t\sps 1$ depends only on $\va_{[t - 1]}$, and \eqref{eq:one_reg_bound_context} holds with
\[\reg\sps 1(T) \defeq 20\sqrt{T\log\Par{\frac{4}{\delta\eps}}},\]
with probability $\ge 1 - \delta$, where for each $t \in [T]$, we require one i.i.d.\ draw $(\vx_t, y_t) \sim \calD$.
\end{lemma}
\begin{proof}
More concretely, our goal in \eqref{eq:one_reg_bound_context} is to have
\[ \sup_{\vu \in \Delta^{\calN}}\sum_{t \in [T]} \inprod{\E_{(\vx, y) \sim \calD}\Brack{\vv\sps 1(\va_t)(\vx, y)}}{\vu - \vu_t} \le \reg\sps 1(T).\]
For any $\va_t$ we have an unbiased estimator of $\E_{(\vx, y) \sim \calD}[\vv\sps 1(\va_t)(\vx, y)]$ conditioned on the history, with entries always in $[-1, 1]$, under our assumed sample access to $\calD$. Thus, the conclusion follows by applying the variant of multiplicative weights in Lemma~\ref{lem:high_prob_sgd} with $L = 1$ and $\Theta = \log(\frac 1 \eps + 2)$.
\end{proof}

\begin{lemma}\label{lem:stat_mloo}
Following definitions \eqref{eq:adef_binary_stat}, \eqref{eq:uvdef_binary_stat}, and assuming \eqref{eq:c_l_bound}, there exists $\calO$, an $\eps$-CMLOO.
\end{lemma}
\begin{proof}
The CMLOO is unchanged from Lemma~\ref{lem:oracle_binary}, except we now return a function $\va \in \calH\sps 1$ such that for any $\vx \in \calX$, we let $\va(\vx)$ have the same output as Algorithm~\ref{alg:oracle_binary} with context $\vx$. Then, for any auxiliary inputs $\vw = (q, r) \in \Delta^2$, $\vu \in \Delta^{\Abs{\calN}+1}$, $c \in \calC$, and $\ell \in \calL$, we have for all $(\vx, y) \in \xset \times [0, 1]$,
\begin{gather*}q\Par{\sum_{s \in \calN} \vu_s \E_{p \sim \va(\vx)}\Brack{(p - y)\sign(p-s)} }
+ r\dd_\ell(c(\vx))\E_{p \sim \va(\vx)}[p - y] \le \eps. \end{gather*}
Taking expectations over the above display over $(\vx, y) \sim \calD$ then yields the CMLOO guarantee
\[\E_{(\vx, y) \sim \calD}\Brack{q \inprod{\vu}{\vv\sps 1(\va)} + r\dd_\ell(c(\vx))\E_{p \sim \va(\vx)}[p - y]} \le 0. \]
\end{proof}

We can now derive the statistical analog of Corollary~\ref{cor:online_bi}, again postponing the discussion of online learners meeting \eqref{eq:ma_regret_stat} to Sections~\ref{ssec:binary_linear} and~\ref{ssec:binary_general}.

\begin{corollary}[Statistical binary omniprediction]\label{cor:offline_bi}
Let $\calL$ be a family of loss functions and $\calC$, $\calC'$ be families of comparators satisfying \eqref{eq:c_l_bound}. 
Assume there exists an online learner $\alg\sps 2$ that takes inputs $\{v_t: \calX \times \{0, 1\} \to \R\}_{t \in [T]}$, and outputs $\ell_{[T]} \in \calL^T$, $c_{[T]} \in (\calC')^T$, such that $(\ell_t, c_t)$ depends only on $v_{[t - 1]}$, and
\begin{equation}\label{eq:ma_regret_stat}
\sup_{(\ell, c) \in \calL \times \calC} \sum_{t \in [T]} \inprod{v_t}{\dd_\ell(c) - \dd_{\ell_t}(c_t)} \le \reg(T),
\end{equation}
for $\reg: \N \to \R_{\ge 0}$ such that all $T \ge T_{\calL, \calC}$ satisfy $\frac{\reg(T)}{T} \le \eps$ with probability $\ge 1 - \frac \delta 2$.
Then if $T = \Omega(\frac{1}{\eps^2} \log(\frac 1{\delta\eps})) + T_{\calL, \calC}$, we can produce $\vp: \calX \to [0, 1]$, a $15\eps$-omnipredictor for $(\calD, \calL, \calC)$, with probability $\ge 1 - \delta$, given $T$ i.i.d.\ samples from $\calD$. 
\end{corollary}
\begin{proof}
The proof is completely analogous to Corollary~\ref{cor:online_bi} but using \eqref{eq:ma_regret_stat} and Lemmas~\ref{lem:offline_regret_bi_cal},~\ref{lem:stat_mloo} in place of \eqref{eq:ma_regret} and Lemmas~\ref{lem:online_regret_bi_cal},~\ref{lem:oracle_binary}. The conclusion is that we can return $\{\vp_t: \calX \to [0, 1]\}_{t \in [T]}$ satisfying $\calF$-multiaccuracy and $\calW$-calibration on average, i.e., for all $f \in \calF$ and $w \in \calW$ in Proposition~\ref{prop:omni},
\begin{align*}
\E_{(\vx, y) \sim \calD}\Brack{\frac 1 T \sum_{t \in [T]} \inprod{\vp_t(\vx) - \vy}{\vf(\vx)} } \le 3\eps , \\
\E_{(\vx, y) \sim \calD}\Brack{\frac 1 T \sum_{t \in [T]} \inprod{\vp_t(\vx) - \vy}{\vw(\vp_t(\vx))}}\le 12\eps,
\end{align*}
simultaneously except with probability $\delta$ (taking a union bound over \eqref{eq:ma_regret_stat} and Lemma~\ref{lem:offline_regret_bi_cal} with $\delta \gets \frac \delta 2$). By linearity of expectation, outputting $\vp \gets \vp_t$ for a uniformly random $t \in [T]$ satisfies $\calF$-multiaccuracy and $\calW$-calibration, so applying Proposition~\ref{prop:omni} to this $\vp$ gives the result.
\end{proof}

\subsection{Generalized linear models}\label{ssec:binary_linear}

In this section, we specialize Corollary~\ref{cor:online_bi} and~\ref{cor:offline_bi} to the setting of generalized linear models, where $\calL \defeq \Lglm$ as defined in \eqref{eq:glm_def}, and $\calC \defeq \Clin$ where
\begin{equation}\label{eq:c_lin_def}
\Clin \defeq \Brace{c(\vx) \defeq \inprod{\vc}{\vx} \mid \vc \in \ball_2^d(1)}.
\end{equation}
We conflate the actual linear classifier $\vc \in \ball^d_2(1)$ with a function $c: \calX \to [-1, 1]$ by using boldface, so e.g., $\vc_t \in \ball^d_2(1)$ corresponds to the function $c_t = \inprod{\vc_t}{\cdot}\in \Clin$.
Recalling \eqref{eq:glm_dl}, we take the discrete derivative $\dd_\ell$ to be negation for all $\ell \in \Lglm$, so $\calF$ in Proposition~\ref{prop:omni} is equivalent to $\Clin$ because $\Clin$ is closed under negation. We next require online learners satisfying \eqref{eq:ma_regret}, \eqref{eq:ma_regret_stat}.

\begin{lemma}\label{lem:online_regret_bi_mul_linear}
Assuming \eqref{eq:c_l_bound} holds, there exists $\alg\sps 2$ such that for any $(v_{[T]}, \vx_{[T]}) \in [-1, 1]^T \times \calX^T$, $\alg\sps 2$ outputs $\vc_{[T]} \in (\ball_2^d(1))^T$ in $O(dT)$ time, such that $\vc_t$ only depends on $v_{[t - 1]}, \vx_{[t - 1]}$, and
\[\sup_{c \in \Clin} \sum_{t \in [T]} v_t \inprod{\vx_t}{\vc - \vc_t} \le \sqrt{T}. \]
\end{lemma}
\begin{proof}
This follows from Lemma~\ref{lem:mirror} with $\calX \gets \ball_2^d(1)$ and $r(\vc) \defeq \half \norm{\vc}_2^2$, where we take $\vg_t \defeq -v_t\vx_t$ so that our application satisfies $\norm{\vg_t}_2 \le L \defeq 1$ and $\Theta \le \half$, because $\norm{\vx_t}_2 \le 1$ under \eqref{eq:c_l_bound}.
\end{proof}

\begin{lemma}\label{lem:offline_regret_bi_mul_linear}
Let $\delta \in (0, 1)$. Assuming \eqref{eq:c_l_bound} holds, there exists $\alg\sps 2$ such that for any $\{v_t: \calX \times \{0, 1\} \to [-1, 1]\}_{t \in [T]}$, $\alg\sps 2$ outputs $\vc_{[T]} \in (\ball_2^d(1))^T$ in $O(dT)$ time, such that $\vc_t$ only depends on $v_{[t - 1]}$, and
\[\sup_{c \in \Clin} \sum_{t \in [T]} \E_{(\vx, y) \sim \calD}\Brack{ v_t(\vx, y) \inprod{\vx}{\vc - \vc_t}} \le 20\sqrt{T\log\Par{\frac 2 \delta}},\]
with probability $\ge 1 - \delta$, where for each $t \in [T]$, we require one i.i.d.\ draw $(\vx_t, y_t) \sim \calD$.
\end{lemma}
\begin{proof}
The proof is identical to Lemma~\ref{lem:online_regret_bi_mul_linear}, where we use Lemma~\ref{lem:high_prob_sgd} in place of Lemma~\ref{lem:mirror}, granting us unbiased access to $\vg_t \defeq \E_{(\vx, y) \sim \calD}[v_t(\vx, y)\vx]$ under our sampling assumption.
\end{proof}

We now combine the pieces to give our result on omnipredicting (binary) generalized linear models.

\begin{theorem}[Binary generalized linear models]\label{thm:glm_bi}
Let $\delta \in (0, 1)$, let $\calL \defeq \Lglm$ and $\calC \defeq \Clin$ defined in \eqref{eq:glm_def_2}, \eqref{eq:c_lin_def} respectively, and assume \eqref{eq:c_l_bound} holds. Then if
\[T = \Omega\Par{\frac{\log\Par{\frac 1 {\delta\eps}}}{\eps^2}}\]
for an appropriate constant, in the online setting, we can output $p_{[T]} \in [0, 1]^T$, an $\eps$-omnipredictor for $(\vx_{T}, y_{[T]}, \calL, \calC)$, in time $O((d + \frac 1 \eps)T)$ with probability $\ge 1 - \delta$. In the statistical setting, we can output $\vp: \calX \to [0, 1]$, an $\eps$-omnipredictor for $(\calD, \calL, \calC)$, in time $O((d +  \frac 1 \eps)T)$ with probability $\ge 1 - \delta$, given $T$ i.i.d.\ samples from $\calD$, such that $\vp$ can be evaluated in time $O(d + \frac 1 \eps)$.
\end{theorem}
\begin{proof}
The result on online omniprediction is immediate by combining Corollary~\ref{cor:online_bi} and Lemma~\ref{lem:online_regret_bi_mul_linear}, and adjusting $\eps \gets \frac \eps {15}$. For the result on statistical omniprediction, the result similarly follows from Corollary~\ref{cor:offline_bi} and Lemma~\ref{lem:offline_regret_bi_mul_linear}. Note that we take $\calC = \calC'$ in these applications.

To evaluate $\vp$ as specified in Corollary~\ref{cor:offline_bi},\footnote{We preprocess the indices in $[T]$ so a uniform sample is attainable in $O(1)$ time, e.g., via the alias method.} we store the values of all inputs to Algorithm~\ref{alg:csba} for each iteration. We can compute
the function $h(s)$ in \eqref{eq:cmloo_helper} for all $s \in \calN$ in $O(\frac 1 \eps)$ time, after spending $O(d)$ time to evaluate some $c_t(\vx)$. This complexity also dominates the cost of each iteration.
\end{proof}

\subsection{General classifiers and losses}\label{ssec:binary_general}

We finally consider Corollary~\ref{cor:online_bi} and Corollary~\ref{cor:offline_bi} for general loss functions $\calL$ and general comparators $\calC$ that satisfy \eqref{eq:c_l_bound}. To state the guarantees of our online learners satisfying \eqref{eq:ma_regret}, \eqref{eq:ma_regret_stat}, we define the following complexity measure parameters of a function class $\calF$. 

\begin{definition}[Statistical Rademacher complexity]
Let $\calF$ be a class of functions $f:\calX \to \R$, $\calD$ be a distribution over $\calX$, and $T \in \N$. The \emph{statistical Rademacher complexity} of $\calF$ is defined to be
\[
\rad_T(\calF) : = \E_{\{\vx_t\}_{t \in [T]}\simiid \calD}\Brack{  \E_{\sigma_{[T]} \simu \{\pm 1\}^T}\left[\sup_{f\in \calF}\frac 1T\sum_{t \in [T]}\sigma_t f(\vx_t)\right]}.
\]
\end{definition}

\begin{definition}[Sequential Rademacher complexity]
Let $\calF$ be a class of functions $f:\calX \to \R$ and $T \in \N$. The \emph{sequential Rademacher complexity} of $\calF$ is defined to be
\[
\srad_T(\calF) : = \sup_{\{\vx_t: \{\pm 1\}^{t - 1} \to \calX\}_{t \in [T]}}  \E_{\sigma_{[T]} \simu \{\pm 1\}^T}\left[\sup_{f\in \calF}\frac 1T\sum_{t \in [T]}\sigma_t f(\vx_t(\sigma_{[t - 1]}))\right].
\]
\end{definition}

For the online setting, we use the following result from~\cite{OkoroaforKK25}.

\begin{lemma}[Theorem 4.5, \cite{OkoroaforKK25}]\label{lem:online_regret_bi_mul_general}
In the setting of Corollary~\ref{cor:online_bi}, let $\calF\defeq \{\dd_{\ell}\circ c\}_{\ell\in\calL, c\in \calC}$. There exists $\alg\sps 2$ such that for any $(v_{[T]}, \vx_{[T]}) \in [-1, 1]^T \times \calX^T$, $\alg\sps 2$ outputs $c_{[T]}$, such that $c_t\in\calC'$ only depends on $v_{[t - 1]}, \vx_{[t - 1]}$, and
\[\sup_{c \in \calC} \sum_{t \in [T]} v_t\Par{\dd_\ell (c(\vx_t)) - \dd_{\ell_t} (c_t(\vx_t))} \le 2T\cdot\srad_T\Par{\calF}. \]
\end{lemma}

We remark that Lemma~\ref{lem:online_regret_bi_mul_general} is based on a computationally-inefficient (indeed, nonconstructive) argument from \cite{RakhlinST15}, but that the regret bound is known to be tight up to a constant factor (Theorem 4.5, \cite{OkoroaforKK25}). For specific pairs $(\calL, \calC)$, e.g., the ones in Lemma~\ref{lem:online_regret_bi_mul_linear}, it is possible to design more explicit online learners, so in general the computational cost depends on the setting. 

For the statistical setting, we similarly use the following result.

\begin{lemma}[Lemma 7.4, Lemma 7.6, \cite{OkoroaforKK25}]\label{lem:offline_regret_bi_mul_general}
In the setting of Corollary~\ref{cor:offline_bi}, let $\delta \in (0, 1)$ and $\calF\defeq \{\dd_{\ell}\circ c\}_{\ell\in\calL, c\in \calC}$. Let $\calV\subseteq \{v:\calX \times \{0, 1\} \to [-1, 1]\}$. There exists $\alg\sps 2$ such that for any $v_{[T]} \in \calV^T$, $\alg\sps 2$ outputs $c_{[T]}$, making $O(T^{1.5})$ calls to an ERM oracle for $\calF$ over $T$ samples per iteration, such that $c_t \in \calC'$ only depends on $v_{[t - 1]}$, and for a universal constant $C$,
\[\sup_{c \in \calC} \sum_{t \in [T]} \E_{(\vx, y) \sim \calD}\Brack{ v_t(\vx, y) \Par{\dd_\ell (c(\vx))-\dd_{\ell_t}  ( c_t(\vx) )}} \le C\Par{\sqrt{T\cdot\log\frac{1}{\delta}}+T\cdot\rad_T\Par{\calF\cdot\calV}},\]
 with probability $\ge 1 -\delta$, where for each $t \in [T]$, we require $T$ i.i.d.\ draws $\sim \calD$. 
\end{lemma}

In the statement of Lemma~\ref{lem:offline_regret_bi_mul_general}, we let $\calF\cdot \calV$ consist of functions $(\vx,y)\mapsto f(\vx)v(\vx,y)$ for $f\in \calF$ and $v\in \calV$, and an ERM oracle for $\calF$ finds $\min_{f \in \calF}  \frac 1 n \sum_{i \in [n]} \vw_i f(\vx_i, y_i)$ over some dataset $\{(\vx_i, y_i)\}$ of $n$ i.i.d.\ draws from $\calD$, and some weights $\vw \in \R^n$. The computational complexity of implementing such an oracle again typically depends on the specific setting.

Finally, we conclude with our main result for binary omniprediction in the general case.

\begin{theorem}[General binary omniprediction]\label{thm:gbm_bi}
Let $\calL$ be a family of loss functions and $\calC$ be a family of comparators such that \eqref{eq:c_l_bound} holds, let $\calF \defeq \{\dd_\ell \circ c\}_{\ell \in \calL, c \in \calC}$, and let $\delta \in (0, 1)$. Let $\TLC^{\textup{seq}}$, $\TLC^{\textup{stat}}$ be such that all $T \ge \TLC^{\textup{seq}}$ satisfy $\srad_T(\calF) \le \frac \eps {30}$, and all $T \ge \TLC^{\textup{stat}}$ satisfy $\rad_T(\calF\cdot \calV) \le \frac \eps {30C}$, where $\calV$ consists of functions $(\vx,y)\to p(\vx) - y$ for all possible $p:\calX\to [0,1]$ outputted by the CMLOO in \Cref{lem:stat_mloo} given classes $\calC,\calL$.
Then if 
\[T = \Omega\Par{\frac{\log\Par{\frac 1 {\delta\eps}}}{\eps^2}} + T^{\textup{seq}}_{\calL, \calC},\]
for an appropriate constant, in the online setting, we can output $p_{[T]} \in [0, 1]^T$, an $\eps$-omnipredictor for $(\vx_{T}, y_{[T]}, \calL, \calC)$, with probability $\ge 1 - \delta$. 
If
\[T = \Omega\Par{\frac{\log\Par{\frac 1 {\delta\eps}}}{\eps^2}} + T^{\textup{stat}}_{\calL, \calC},\]
for an appropriate constant, in the statistical setting, we can output  $\vp: \calX \to [0, 1]$, an $\eps$-omnipredictor for $(\calD, \calL, \calC)$, with probability $\ge 1 - \delta$, given $T^2$ i.i.d.\ samples from $\calD$ and $O(T^{2.5})$ calls to an ERM oracle for $\calF$.
\end{theorem}
\begin{proof}
The proof is largely the same as the proof for Theorem~\ref{thm:glm_bi}. We apply Lemma~\ref{lem:online_regret_bi_mul_general} within Corollary~\ref{cor:online_bi} for the online setting and Lemma~\ref{lem:offline_regret_bi_mul_general} within Corollary~\ref{cor:offline_bi} in the statistical setting. The oracle complexity comes from the cost of the online learner in Lemma~\ref{lem:offline_regret_bi_mul_general}.
\end{proof}
\section{Multiclass Omniprediction }\label{sec:okk}

We now proceed to our main result: omniprediction with $k > 2$ classes. Here, many properties specific to the binary classification setting do not hold, e.g., the generalizations of Lemmas~\ref{lem:thresh_suffice} and~\ref{lem:oracle_binary}. We develop a general strategy for constructing MLOOs for multiclass omniprediction in Section~\ref{ssec:mixture_oracle}. We then give the multiclass extensions of Corollaries~\ref{cor:online_bi} and~\ref{cor:offline_bi} in Section~\ref{ssec:multi_omni}. Finally, in Sections~\ref{ssec:multi_linear} and~\ref{ssec:multi_general}, we give our full multiclass omniprediction results for linear and general classifiers.

Throughout, we make the following normalization assumptions:
\begin{equation}\label{eq:cl_multi}
\begin{gathered}
\calX \subseteq \ball^d_2(1),\quad \vc(\vx) \in [-1, 1]^k \text{ for all } \vc \in \calC,\; \vx \in \calX, \\
\ell(\vt, \vy) \in [-1, 1] \text{ for all } \ell \in \calL,\; (\vt, \vy) \in [-1, 1]^k \times \partial \Delta^k.
\end{gathered}
\end{equation}

\subsection{MLOOs for multiclass omniprediction}\label{ssec:mixture_oracle}

In this section, we consider a specialized application of the machinery in Section~\ref{ssec:framework} to multiclass prediction. Specifically, suppose that we have an instance of Problem~\ref{prob:sba}, where
\begin{equation}\label{eq:net_setting}\calA \defeq \Delta^{\calN},\quad \calB = \partial \Delta^k,\end{equation}
and $\calN$ is an $\eps$-net for $\Delta^k$. Also, define for all $\va \in \Delta^{\calN}$ and $\vb \in \calB$,
\begin{equation}\label{eq:base_v_def}
\vv\Par{\va, \vb} \defeq \Brace{\va_{\vs}\Par{\vs - \vb}}_{\vs \in \calN} \in \R^{k|\calN|},
\end{equation}
and suppose that for all $i \in [m]$,
\begin{equation}\label{eq:vvi}\vv^{(i)}\Par{\va, \vb} = \mm^{(i)} \vv\Par{\va, \vb}\end{equation}
for some linear operator $\mm^{(i)}: \R^{k|\calN|} \to \calH^{(i)}$. We give a meta-result that shows how to implement an MLOO for arbitrary simultaneous Blackwell approachability instances satisfying \eqref{eq:base_v_def}, \eqref{eq:vvi}, whose quality scales with bounds on the $\{\mm^{(i)}\}_{i \in [m]}$ and the $\{\calU^{(i)}\}_{i \in [m]}$.

\begin{lemma}\label{lem:minimax}
In the setting of Problem~\ref{prob:sba}, suppose \eqref{eq:net_setting}, \eqref{eq:base_v_def}, and \eqref{eq:vvi} hold, where $\calN$ is an $\eps$-net for $\Delta^k$. Further, suppose that for all $i \in [m]$, we have
\begin{equation}\label{eq:m_op}\norm{\Par{\mm^{(i)}}^* \vu^{(i)}}_\infty \le R \text{ for all } \vu^{(i)} \in \calU^{(i)}.\end{equation}
where $^*$ denotes the adjoint. We can implement a $2\eps R$-MLOO with probability at least $1-\delta$ in time $O(|\calN|\cdot\poly(k, \log \frac 1 {\delta\eps}))$. Similarly, in the setting of Problem~\ref{prob:csba}, we can implement a $2\eps R$-CMLOO with probability at least $1-\delta$ in time $O(|\calN|\cdot\poly(k, \log \frac 1 {\delta\eps}))$.
\end{lemma}
\begin{proof}
We first show that for all $\vw \in \Delta^m$, $\{\vu^{(i)}\}_{i \in [m]} \in \prod_{i \in [m]} \calU^{(i)}$, there exists $\va \in \calA$ with
\[\max_{\vb \in \calB} \sum_{i \in [m]} \vw_i \inprod{\vu^{(i)}}{\vv^{(i)}(\va, \vb)} \le \eps R.\]
Throughout the proof fix a set of $\vw \in \Delta^m$ and $\{\vu^{(i)}\}_{i \in [m]} \in \prod_{i \in [m]} \calU^{(i)}$, and denote
\[\vf \defeq \sum_{i \in [m]} \Par{\mm^{(i)}}^* \vu^{(i)} \in \ball_\infty^{k|\calN|}\Par{R}. \]
Thus, our goal is to establish
\begin{equation}\label{eq:opt_problem}\min_{\va \in \calA} \max_{\vb \in \calB} \inprod{\vf}{\vv(\va, \vb)} \le \eps R.\end{equation}
Because $\inprod{\vf}{\vv(\va, \vb)}$ is a bilinear function of $\va, \vb$, the von Neumann minimax theorem gives
\begin{align*}\min_{\va \in \calA} \max_{\vb \in \calB} \inprod{\vf}{\vv(\va, \vb)} &= \max_{\vq \in \Delta^k} \min_{\vs \in \calN} \E_{\vb \sim \vq}\Brack{\inprod{\vf}{\vv(\ve_{\vs}, \vb)}} \\
&= \max_{\vq \in \Delta^k} \min_{\vs \in \calN} \E_{\vb \sim \vq}\Brack{\inprod{\vf_{\vs}}{\vs - \vb}} = \max_{\vq \in \Delta^k} \min_{\vs \in \calN} \inprod{\vf_{\vs}}{\vs - \vq},
\end{align*}
where $\ve_{\vs} \in \{0, 1\}^\calN$ is the indicator vector for strategy $\vs \in \calN$, and $\vf_{\vs} \in \ball_\infty^k(R)$ concatenates the corresponding coordinates of $\vf$. Finally we claim that for any $\vq \in \Delta^k$,
\[\min_{\vs \in \calN} \inprod{\vf_{\vs}}{\vs - \vq} \le \eps R.\]
Indeed, choosing $\vs \in \calN$ so that $\norm{\vs - \vq}_1 \le \eps$ and applying H\"older's inequality yields this bound.

We conclude by discussing runtime. Normalize the problem by $R$ by resetting $\vf \gets \frac 1 R \vf$, so we want to solve \eqref{eq:opt_problem} to $\eps$ additive error. Notice that \eqref{eq:opt_problem} is of the following form:
\[\min_{\va \in \Delta^{\calN}} \max_{\vb \in \Delta^k} \vg^\top \va - \vb^\top \mf \va = \min_{\va \in \Delta^{\calN}} \max_{\vb \in \Delta^k} \vb^\top \mm \va, \]
where $\mf \in \R^{k \times \calN}$ horizontally stacks the values of $\vf$, $\vg \in \R^{\calN}$ has coordinate $\vs \in \calN$ equal to $\inprod{\vf_{\vs}}{\vs}$, and we define $\mm \defeq \1_k\vg^\top - \mf$. Also, we have that $\mm \in [-2, 2]^{k \times \calN}$. We can rewrite this as
\begin{align*}
\min t \text{ such that } \ma \va + \vc = t\1_k,\; \1_{\calN}^\top \va = 1,\; \va \ge \0_{\calN},\; \vc \ge \0_k \text{ entrywise.}
\end{align*}
We note that $\vc$ is enforcing the inequality constraints $\ma\va \le t\1_k$. We can trivially enforce that $t \in [-2, 2]$, $\va \in [0, 1]^{\calN}$, and $\vc \in [0, 4]^k$. It is enough to obtain $\eps$ additive error for this problem for our guarantees. At this point, the solver in Theorem 1.1 of \cite{BrandLLSSSW21} gives the claim. 
\end{proof}

\subsection{Reducing multiclass omniprediction to low-regret learning}\label{ssec:multi_omni}

We give the analogs of Corollaries~\ref{cor:online_bi} and~\ref{cor:offline_bi} in the multiclass setting.

\textbf{Online setting.} In the online setting, for a fixed parameter $\eps \in (0, 1)$, we define $(\calA, \calB)$ as in \eqref{eq:net_setting},
where $\calN$ is an $\eps$-net for $\Delta^k$ of size $(\frac 5 \eps)^{k - 1}$ as guaranteed by Fact~\ref{fact:net_l1}. For some $\va \in \calA$, we use $\vp \sim \va$ to mean that some $\vp \in \calN$ is sampled according to $\va$. 

We next define the sets and payoff vectors in Problem~\ref{prob:csba}:
\begin{equation}\label{eq:uvdef_multi_online}
\begin{gathered}\calU^{(1)} \defeq [-1, 1]^{\calN \times k},\quad \vv^{(1)}(\va, \vb) \defeq \Brace{\va_{\vs} (\vs - \vb)}_{\vs \in \calN}, \\
\calU^{(2)} \defeq \Brace{\vd_\ell \circ \vc}_{\ell \in \calL, \vc \in \calC},\quad \vv^{(2)}(\va, \vb) \defeq \E_{\vp \sim \va}\Brack{\vp-\vb}.
\end{gathered}
\end{equation}
Note that $\calU\sps 1$ and $\vv\sps 1$ live in a vector space of dimension $k|\calN|$, whereas $\calU\sps 2$ and $\vv\sps 2$ are functions with range in $\R^k$.  
We require the analog of Lemma~\ref{lem:online_regret_bi_cal},  an online learner for $\calU\sps 1$.

\begin{lemma}\label{lem:online_regret_multi_cal}
Following definitions \eqref{eq:net_setting}, \eqref{eq:uvdef_multi_online}, there exists $\alg\sps 1$ such that for any $(\va_{[T]}, \vb_{[T]}) \in \calA^T \times \calB^T$, $\alg\sps 1$ outputs $\vu\sps 1_{[T]} \in (\calU\sps 1)^T$ such that $\vu_t\sps 1$ depends only on $\va_{[t - 1]}$, $\vb_{[t - 1]}$, and \eqref{eq:one_reg_bound_context} holds with
\[\reg\sps 1(T) \defeq \eps T + \frac{k|\calN|} \eps.\]
\end{lemma}
\begin{proof}
The algorithm is projected gradient descent. More precisely, because $\norm{\vs - \vb}_1 \le 2$ for all $(\vs, \vb) \in \calN \times \calB$, we have for all $(\va_t, \vb_t) \in \calA \times \calB$, that
\[\norms{\vv\sps 1(\va_t, \vb_t)}_2 \le \norms{\vv\sps 1(\va_t, \vb_t)}_1 \le 2.\]
Therefore, standard regret analyses of projected gradient descent with step size $\eta \gets \frac \eps 2$ (e.g., Theorem 3.2, \cite{Bubeck15}) gives the result, because $\calU\sps 1$ has $\ell_2$ radius at most $\sqrt{k|\calN|}$.
\end{proof}

\begin{corollary}[Online multiclass omniprediction]\label{cor:online_multi}
Let $\calL$ be a family of loss functions and $\calC$, $\calC'$ be families of comparators satisfying \eqref{eq:cl_multi}. 
Assume there exists an online learner $\alg\sps 2$ that takes inputs $(\vv_{[T]}, \vx_{[T]}) \in (\ball_2^k(2))^T \times \calX^T$, and outputs $\ell_{[T]} \in \calL^T$, $\vc_{[T]} \in (\calC')^T$, such that $(\ell_t, \vc_t)$ depends only on $\vv_{[t - 1]}$, $\vx_{[t - 1]}$, and
\begin{equation}\label{eq:ma_multi_regret}
\sup_{(\ell, \vc) \in \calL \times \calC} \sum_{t \in [T]} \inprod{\vv_t}{\vd_\ell(\vc(\vx_t)) - \vd_{\ell_t}(\vc_t(\vx_t))} \le \reg(T),
\end{equation}
for $\reg: \N \to \R_{\ge 0}$ such that all $T \ge T_{\calL, \calC}$ satisfy $\frac{\reg(T)}{T} \le \eps$.
Then if $T = \Omega(k(\frac 1 \eps)^{k + 1} + \frac{1}{\eps^2} \log(\frac 1{\delta})) + T_{\calL, \calC}$, we can produce $p_{[T]} \in [0, 1]^T$, a $12\eps$-omnipredictor for $(\vx_{[T]}, \vy_{[T]}, \calL, \calC)$, with probability $\ge 1 - \delta$.
\end{corollary}
\begin{proof}
The proof is completely analogous to Corollary~\ref{cor:online_bi}. We substitute Lemma~\ref{lem:online_regret_multi_cal} and \eqref{eq:ma_multi_regret} for Lemma~\ref{lem:online_regret_bi_cal} and \eqref{eq:ma_regret}, and note that we may take $L = 2$ in \eqref{eq:width_bound_context} by the $\ell_\infty$-$\ell_1$ H\"older's inequality. We postpone discussion of implementing the CMLOO for a moment, but suppose we have a $2\eps$-CMLOO. Then, Corollary~\ref{cor:context-sba} yields a sequence $\vp_{[T]} \in (\Delta^k)^T$ such that with probability $\ge 1 - \delta$,
\begin{equation}\label{eq:calma_multi_online}
\sup_{\vu\sps i \in \calU\sps i} \frac 1 T \sum_{t \in [T]} \inprod{\vu\sps i(\vx_t)}{\vv\sps i(\ve_{\vp_t}, \vb_t)} \le 2\eps + \frac{\reg\sps i(T)+28\sqrt{T\log(\frac 4 \delta)}}{T}
\end{equation}
for $i \in [2]$. When $i = 1$, the guarantee in \eqref{eq:calma_multi_online} corresponds to calibration against the entire $\ell_\infty$-norm ball in dimension $k|\calN|$, which encompasses $\calW$-calibration for $\calW$ in Proposition~\ref{prop:omni}, under the scaling assumption \eqref{eq:cl_multi}. When $i = 2$, the guarantee in \eqref{eq:calma_multi_online} corresponds to $\calF$-multiaccuracy as required by Proposition~\ref{prop:omni}. For large enough $T$ as specified, we thus have $5\eps$-$\calW$-calibration using Lemma~\ref{lem:online_regret_multi_cal} as $\alg\sps 1$, and $4\eps$-$\calF$-multiaccuracy using \eqref{eq:ma_multi_regret}, and Proposition~\ref{prop:omni} gives the claim.

It remains to give a $2\eps$-CMLOO. For this we use Lemma~\ref{lem:minimax}. Comparing the definitions \eqref{eq:vvi} and \eqref{eq:uvdef_multi_online}, $\mm\sps 1$ is simply the identity matrix in dimension $k|\calN|$, and $\mm\sps 2$ is $\1_k \otimes \id_{\calN}$, where $\otimes$ denotes the Kronecker product. This matrix has one-sparse columns, so it  satisfies
\[\norm{\mm\sps 2}_{1 \to 1} = 1 \implies \norm{\Par{\mm\sps 2}^*}_{\infty \to \infty} = 1.\]
Hence, we may take $R = 1$ in \eqref{eq:uvdef_multi_online}, because both $\calU^{\sps 1}$ and $\calU^{\sps 2}$ are contained in the $\ell_\infty$ balls of their respective dimension. The result now follows from Lemma~\ref{lem:minimax}.

\end{proof}

\textbf{Statistical setting.} We let $\calH\sps 1$ and $\calH\sps 2$ be the Hilbert spaces of (norm) square-integrable functions under $\calD$, with ranges $\R^{\calN \times k}$, $ \R^k$, respectively, with the standard $L^2(\calD)$ inner products.

Next, we take $\calA$ to be functions  taking each $\vx \in \calX \to \va(\vx) \in \Delta^{\calN}$, i.e.,
\begin{equation}\label{eq:adef_multi_stat}\calA \defeq \Brace{\va: \calX \to \Delta^{\calN} }.\end{equation}
Our payoff vectors will again be independent of $\vb \in \calB$, so we omit it from our notation. Also, let
\begin{equation}\label{eq:uvdef_multi_stat}
\begin{aligned}
\calU\sps 1 &\defeq \Delta^{\calN},\quad \vv\sps 1(\va)(\vx, \vy) \defeq \Brace{[\va(\vx)]_{\vs}(\vs - \vy)}_{\vs \in \calN}, \\
\calU\sps 2 &\defeq \Brace{\vd_\ell \circ \vc}_{\ell \in \calL, \vc \in \calC}, \quad \vv\sps 2(\va)(\vx, \vy) \defeq \E_{\vp \sim \va(\vx)}[\vp - \vy].
\end{aligned}
\end{equation}
We last require an online learner for $\calU\sps 1$ in the statistical setting.

\begin{lemma}\label{lem:offline_regret_multi_cal}
Following definitions \eqref{eq:adef_multi_stat}, \eqref{eq:uvdef_multi_stat}, there exists $\alg\sps 1$ such that for any $\va_{[T]} \in \calA^T$, $\alg\sps 1$ outputs $\vu\sps 1_{[T]} \in (\calU\sps 1)^T$ such that $\vu_t\sps 1$ depends only on $\va_{[t - 1]}$, and \eqref{eq:one_reg_bound_context} holds with
\[\reg\sps 1(T) \defeq \eps T + \frac{10k|\calN|}{\eps} + 32\sqrt{T\log\Par{\frac 2 \delta}},\]
with probability $\ge 1 - \delta$, where for each $t \in [T]$, we require one i.i.d.\ draw $(\vx_t, \vy_t) \sim \calD$.
\end{lemma}
\begin{proof}
We pattern our proof off of Lemma~\ref{lem:high_prob_sgd}, although we require a few differences to obtain the specific form of regret bound here. The key observation is that for all $\va \in \calA$, the definitions \eqref{eq:uvdef_multi_stat} give $|\inprod{\vu\sps 1}{\vv\sps 1(\va)}| \le 2$ using the $\ell_\infty$-$\ell_1$ H\"older's inequality. Our strategy is then to play (stochastic) projected gradient descent (PGD) against the $\{\vv\sps 1(\va_t)\}_{t \in [T]}$. To simplify notation, let
\[\vg_t \defeq \E_{(\vx, \vy)\sim \calD}\Brack{\vv\sps 1(\va_t)(\vx, \vy)},\quad \tvg_t \defeq \vv\sps 1(\va_t)(\vx_t, \vy_t),\quad \vd_t \defeq \vg_t - \tvg_t, \]
and observe that under our sampling assumptions, $\tvg_t$ is unbiased for $\vg_t$ conditioned on the history of the algorithm if we use a held out i.i.d.\ sample $(\vx_t, \vy_t)$. Also, $|\inprod{\tvg_t}{\vu\sps 1}| \le 2$ holds for all $t \in [T]$ and $\vu\sps 1 \in \calU\sps 1$, and $\max_{t \in [T]} \max\{\norms{\vg_t}_1, \norms{\tvg_t}_1\} \le 2$.

Now, we define $\vu_1 \gets \0_{\calN \times k}$, and our iterates $\vu_t$ using PGD with step size $\eta > 0$ and the $\{-\tvg_t\}_{t \in [T]}$,
\begin{equation}\label{eq:standard_pgd}\vu_t\sps 1 \gets \argmin_{\vu\sps 1 \in \calU\sps 1} \Brace{\norm{\vu\sps 1 - \Par{\vu\sps 1_{t - 1}+ \eta \tvg_{t - 1}}}_2^2}.\end{equation}
We also define a ``ghost iterate'' sequence of $\vw_{[T + 1]} \in (\calU\sps 1)^{T + 1}$ that sets $\vw_1 = \vu\sps 1_1$, but updates using $\vd_{t - 1}$ in place of $\tvg_{t - 1}$ in \eqref{eq:standard_pgd}. Standard PGD analysis (e.g., Theorem 3.2, \cite{Bubeck15}) shows
\begin{align*}
\sum_{t \in [T]} \inprod{\tvg_t}{\vu\sps 1 - \vu_t\sps 1} \le 2\eta T+ \frac{k|\calN|}{2\eta}, \\
\sum_{t \in [T]} \inprod{\vd_{t}}{\vu\sps 1 - \vw_t} \le 8\eta T + \frac{k|\calN|}{2\eta},
\end{align*}
simultaneously hold for all $\vu\sps 1 \in \calU\sps 1$. Summing and rearranging yields
\begin{align*}
\sum_{t \in [T]} \inprod{\vg_t}{\vu\sps 1 - \vu_t\sps 1} \le 10\eta T + \frac{k|\calN|}{\eta} + \sum_{t \in [T]} \inprod{\vd_t}{\vw_t - \vu_t\sps 1}.
\end{align*}
Now the last term above is the sum of $T$ conditionally mean-zero terms, each of which is bounded in $[-8, 8]$. Thus by the Azuma-Hoeffding inequality, with probability $\ge 1 - \delta$,
\[ \sum_{t \in [T]} \inprod{\vg_t}{\vu\sps 1 - \vu_t\sps 1} \le 10\eta T + \frac{k|\calN|}{\eta}+ 32\sqrt{T\log\Par{\frac 2 \delta}},\]
and supremizing this over $\vu\sps 1 \in \calU\sps 1$ and setting $\eta \gets \frac \eps {10}$ gives the claim.
\end{proof}

\begin{corollary}[Statistical multiclass omniprediction]\label{cor:offline_multi}
Let $\calL$ be a family of loss functions and $\calC$, $\calC'$ be families of comparators satisfying \eqref{eq:c_l_bound}. 
Assume there exists an online learner $\alg\sps 2$ that takes inputs $\{\vv_t: \calX \times \partial \Delta^k \to \ball_2^k(2)\}_{t \in [T]}$, and outputs $\ell_{[T]} \in \calL^T$, $\vc_{[T]} \in (\calC')^T$, such that $(\ell_t, \vc_t)$ depends only on $\vv_{[t - 1]}$, and
\begin{equation}\label{eq:ma_regret_stat_multi}
\sup_{(\ell, \vc) \in \calL \times \calC} \sum_{t \in [T]} \inprod{\vv_t}{\vd_\ell(\vc) - \vd_{\ell_t}(\vc_t)} \le \reg(T),
\end{equation}
for $\reg: \N \to \R_{\ge 0}$ such that all $T \ge T_{\calL, \calC}$ satisfy $\frac{\reg(T)}{T} \le \eps$ with probability $\ge 1 - \frac \delta 2$.
Then if $T = \Omega(k(\frac 1 \eps)^{k + 1} + \frac{1}{\eps^2} \log(\frac 1{\delta})) + T_{\calL, \calC}$, we can produce $\vp: \calX \to [0, 1]$, a $9\eps$-omnipredictor for $(\calD, \calL, \calC)$, with probability $\ge 1 - \delta$, given $T$ i.i.d.\ samples from $\calD$. 
\end{corollary}
\begin{proof}
The proof is the same as Corollary~\ref{cor:online_multi} (with modifications analogous to Corollary~\ref{cor:offline_bi} vis-\`a-vis Corollary~\ref{cor:online_bi}), where we use Lemma~\ref{lem:offline_regret_multi_cal} and \eqref{eq:ma_regret_stat_multi} instead of Lemma~\ref{lem:online_regret_multi_cal} and \eqref{eq:ma_multi_regret}. In our construction of the CMLOO in the statistical setting, we note that the matrices $\mm\sps 1$ and $\mm\sps 2$ in \eqref{eq:vvi} are again $\id_{\calN \times k}$ and $\1_k \otimes \id_{\calN}$, so Lemma~\ref{lem:minimax} again yields a $2\eps$-CMLOO that outputs a function $\va_t: \calX \to \Delta^{\calN}$, from which we can sample random predictions $\vp_t: \calX \to \Delta^{\calN}$. 
As in Corollary~\ref{cor:offline_bi}, our final omnipredictor evaluates a uniform randomly sampled $\vp_t$, over the range $t \in [T]$.
\end{proof}

\subsection{Generalized linear models}\label{ssec:multi_linear}

In this section, we specialize Corollaries~\ref{cor:online_multi} and~\ref{cor:offline_multi} to the setting of multiclass generalized linear models, where $\calL \defeq \Lglm$ as defined in \eqref{eq:glm_def}, and $\calC \defeq \Clin$ where
\begin{equation}\label{eq:clin_multi}
\Clin \defeq \Brace{\vc(\vx) \defeq \mc\vx \mid \mc \in \R^{k\times d},\; \norm{\mc}_{2 \to \infty} \le 1}.
\end{equation}
In other words, $\mc$ has sub-unit norm rows. This is the natural family of classifiers because it takes $\vx \in \calX$ to $\vc(\vx) \in [-1, 1]^k$, under the scaling bounds in \eqref{eq:cl_multi}. Analogously to Section~\ref{ssec:binary_linear}, we conflate a function $\vc \in \Clin$ with the associated linear classifier $\mc \in \R^{k \times d}$ via capitalization. We again observe that because $\vd_\ell$ is negation for all $\ell \in \Lglm$ by \eqref{eq:glm_dl}, and $\Clin$ is closed under negation, we can equivalently set $\calF \gets \Clin$ in applications of Proposition~\ref{prop:omni}.

Our last ingredients are online learners satisfying \eqref{eq:ma_multi_regret}, \eqref{eq:ma_regret_stat_multi}.

\begin{lemma}\label{lem:online_regret_multiclass_mul_linear}
Assuming \eqref{eq:cl_multi} holds, there exists $\alg\sps 2$ such that for any $(\vv_{[T]}, \vx_{[T]}) \in (\ball_2^k(2))^T \times \calX^T$, $\alg\sps 2$ outputs $\mc_{[T]} \in (\ball_{2 \to \infty}^{k \times d}(1))^T$ in $O(dkT)$ time, such that $\mc_t$ only depends on $\vv_{[t - 1]}$, $\vx_{[t - 1]}$, and
\[\sup_{\mc \in \Clin}\sum_{t \in [T]} \inprod{\vv_t \otimes \vx_t}{\mc - \mc_t} \le 2\sqrt{kT}.\]
\end{lemma}
\begin{proof}
This follows from Lemma~\ref{lem:mirror} with $\calX \gets \ball^{k \times d}_{2 \to \infty}$ and $r(\mc) \defeq \half \normsf{\mc}^2$ (i.e., half the squared entrywise $\ell_2$ norm). Note that for all $t \in [T]$, because $\vv_t \otimes \vx_t$ is rank-one,
\[\normf{\vv_t \otimes \vx_t} = \normop{\vv_t \otimes \vx_t} = \norm{\vv_t}_2 \norm{\vx_t}_2 \le 2.\]
Thus, the only adjustments compared to Lemma~\ref{lem:online_regret_bi_mul_linear} is that now we have $L = 2$ and $\Theta \le \frac k 2$.
\end{proof}

\begin{lemma}\label{lem:offline_regret_multiclass_mul_linear}
Let $\delta \in (0, 1)$. Assuming \eqref{eq:cl_multi} holds, there exists $\alg\sps 2$ such that for any $\{\vv_t: \calX \times \partial \Delta^k \to \ball_2^k(2)\}_{t \in [T]}$, $\alg\sps 2$ outputs $\mc_{[T]} \in (\ball_{2 \to \infty}^{k \times d}(1))^T$ in $O(dkT)$ time, such that $\mc_t$ only depends on $\vv_{[t - 1]}$, and
\[\sup_{\mc \in \Clin} \sum_{t \in [T]} \E_{(\vx, \vy) \sim \calD}\Brack{\inprod{\vv_t(\vx, \vy) \otimes \vx}{\mc - \mc_t}} \le 40\sqrt{kT\log\Par{\frac 2 \delta}}, \]
with probability $\ge 1 - \delta$, where for each $t \in [T]$, we require one i.i.d.\ draw $(\vx_t, \vy_t) \sim \calD$.
\end{lemma}
\begin{proof}
The proof is identical to Lemma~\ref{lem:online_regret_multiclass_mul_linear}, where we use Lemma~\ref{lem:high_prob_sgd} in place of Lemma~\ref{lem:mirror}.
\end{proof}

We conclude with our main result on omnipredicting multiclass generalized linear models.

\begin{theorem}[Multiclass generalized linear models]\label{thm:glm_multi}
Let $\delta \in (0, 1)$, let $\calL \defeq \Lglm$ and $\calC \defeq \Clin$ defined in \eqref{eq:glm_def}, \eqref{eq:clin_multi} respectively, and assume \eqref{eq:cl_multi} holds. Then if
\[T = k\Par{\Omega\Par{\frac 1 \eps}^{k + 1} + \Omega\Par{\frac{\log\Par{\frac 1 {\delta}}}{\eps^2}}}\]
for an appropriate constant, in the online setting, we can output $\vp_{[T]} \in (\Delta^k)^T$, an $\eps$-omnipredictor for $(\vx_{T}, \vy_{[T]}, \calL, \calC)$, in time $O(dkT) + O(\frac 1 \eps)^{2k}\poly(k, \log \frac 1 {\delta\eps})$ with probability $\ge 1 - \delta$. In the statistical setting, we can output $\vp: \calX \to \Delta^k$, an $\eps$-omnipredictor for $(\calD, \calL, \calC)$, in time $O(dkT) + O(\frac 1 \eps)^{2k}\poly(k, \log \frac 1 {\delta\eps})$ with probability $\ge 1 - \delta$, given $T$ i.i.d.\ samples from $\calD$, such that $\vp$ can be evaluated in time $O(dk) + O(\frac 1 \eps)^{k+1}\poly(k, \log \frac 1 {\delta\eps})$ on any $\vx \in \calX$ with probability $\ge 1 - \delta$.
\end{theorem}
\begin{proof}
For the online omniprediction result, we combine Lemma~\ref{lem:online_regret_multiclass_mul_linear} and Corollary~\ref{cor:online_multi}, and for the statistical omniprediction result, we combine Lemma~\ref{lem:offline_regret_multiclass_mul_linear} and Corollary~\ref{cor:offline_multi}. We note that the runtime cost of each iteration is dominated by the $O(dk)$ time for computing $\vc_t(\vx_t)$, and the cost of Lemma~\ref{lem:minimax}. This also applies to the cost of evaluating $\vp$ on a fresh sample.
\end{proof}

\subsection{General classifiers and losses}\label{ssec:multi_general}

In this section, we specialize Corollaries~\ref{cor:online_multi} and~\ref{cor:offline_multi} to the setting of general multiclass models. 
Analogously to Section~\ref{ssec:multi_linear}, we consider general loss functions $\calL$ and general function class $\calC$ that satisfy \eqref{eq:cl_multi}. We again require online learners satisfying \eqref{eq:ma_multi_regret}, \eqref{eq:ma_regret_stat_multi}.

Our multiclass online learning results apply the binary online learners from Section~\ref{ssec:binary_general} in a black-box way. It is possible that tighter characterizations in the multiclass setting are possible (e.g., in the dependence on $k$), especially for specific structured $(\calC, \calL)$. We demonstrated an example of this in Section~\ref{ssec:multi_linear}, and leave a more general theory to future work. This section is included primarily to highlight how to apply our techniques in a general setting, as our paper's focus is developing the omniprediction framework rather than multiclass learning for specific comparators.

Applying Theorem 4.5 of \cite{OkoroaforKK25} coordinatewise, we obtain the following lemma.

\begin{lemma}\label{lem:online_regret_multiclass_mul_general}
In the setting of \Cref{cor:online_multi}, let $\calF\defeq \{\vd_{\ell}\circ \vc\}_{\ell\in\calL, \vc\in \calC}$. Assuming \eqref{eq:cl_multi} holds for family of loss functions $\calL$ and families of comparators $\calC$ and $\calC'$, there exists $\alg\sps 2$ such that for any $(\vv_{[T]}, \vx_{[T]}) \in (\ball_2^k(2))^T \times \calX^T$, $\alg\sps 2$ outputs $\vc_{[T]} \in \calC'$, such that $\vc_t$ only depends on $\vv_{[t - 1]}$, $\vx_{[t - 1]}$, and
\[\sup_{\vc \in \calC}\sum_{t \in [T]} \inprod{\vv_t }{\vd_{\ell}(\vc(\vx)) - \vd_{\ell_t}(\vc_t(\vx))} \le T\cdot \sum_{i \in [k]}\srad_T(\calF_i),\]
where $\calF_i$ consists of functions $\vx \mapsto [\vf(\vx)]_i$ for $\vf\in \calF$, with $[\vf(\vx)]_i$ being the $i^{\text{th}}$ coordinate of $\vf(\vx)$.
\end{lemma}

Similarly, applying Lemma 7.4 and Lemma 7.6 of \cite{OkoroaforKK25} coordinatewise yields the following.

\begin{lemma}\label{lem:offline_regret_multiclass_mul_general}
In the setting of \Cref{cor:offline_multi}, let $\delta \in (0, 1)$ and $\calF\defeq \{\vd_{\ell}\circ \vc\}_{\ell\in\calL, \vc\in \calC}$. Let $\calV$ be a family of functions $\vv:\calX \times \partial \Delta^k \to \ball_2^k(2)$.
There exists $\alg\sps 2$ such that for any $\vv_{[T]}\in \calV^T$, $\alg\sps 2$ outputs $\vc_{[T]} $, making $O(T^{1.5})$ calls to an ERM oracle for each $\calF_i$ over $T$ samples per iteration, such that $\vc_t\in \calC'$ only depends on $\vv_{[t - 1]}$, and for a universal constant $C$,
\[\sup_{c \in \calC} \sum_{t \in [T]} \E_{(\vx, \vy) \sim \calD}\Brack{\inprod{\vv_t(\vx, \vy)}{\vd_{\ell}(\vc(\vx)) -\vd_{\ell_t}(\vc_t(\vx))}} \le C\left(k\sqrt{T\cdot \log \frac k\delta} + T\cdot \sum_{i \in[k]}\rad_T(\calF_i\cdot \calV_i)\right), \]
with probability $\ge 1 - \delta$, where for each $t \in [T]$, we require $T$ i.i.d.\ draws $(\vx_t, \vy_t) \sim \calD$.
\end{lemma}

In the statement of \Cref{lem:offline_regret_bi_mul_general}, the class $\calF_i\cdot \calV_i$ consists of functions $(\vx,\vy)\mapsto [\vf(\vx)]_i[\vv(\vx,\vy)]_i$ for $\vf\in \calF$ and $\vv\in\calV$, with $[\vf(\vx)]_i,[\vv(\vx,\vy)]_i$ being the $i^{\text{th}}$ coordinates of $\vf(\vx),\vv(\vx,\vy)$, respectively.

We conclude with our main result on multiclass omniprediction in the general setting.

\begin{theorem}[General multiclass omnprediction]\label{thm:gen_multi}
Let $\calL$ be a family of loss functions and $\calC$ be a family of comparators such that \eqref{eq:cl_multi} holds, let $\calF:=\{\vd_{\ell}\circ\vc\}_{\ell\in \calL,\vc\in \calC}$, and let $\delta\in (0,1)$. Let $T_{\calL,\calC}^ \seq, T_{\calL,\calC}^\stat$ be such that all $T\ge T_{\calL,\calC}^\seq$ satisfy $\sum_{i = 1}^k\srad_T(\calF_i) \le \frac \varepsilon 9$, and all $T \ge T_{\calL,\calC}^\stat$ satisfy $\sum_{i = 1}^k \rad_T(\calF_i\cdot \calV_i) \le \frac{\varepsilon}{18C}$, where $\calV$ consists of functions $(\vx,\vy)\to \vp(\vx) - \vy$ for all possible $\vp:\calX\to \Delta^k$ outputted by the CMLOO in \Cref{lem:minimax} given classes $\calC,\calL$. Then if
\[
T = \Omega\left(k\left(\frac 1 \eps\right)^{k + 1} + \frac{\log \frac 1 \delta}{\eps^2} \right) + T_{\calL,\calC}^\seq,
\]
for an appropriate constant, in the online setting, we can output $\vp_{[T]}\in (\Delta^k)^T$, an $\varepsilon$-omnipredictor for $(\vx_{[T]},\vy_{[T]},\calL,\calC)$, with probability $\ge 1-\delta$. If
\[
T = \Omega\left(k\left(\frac 1 \eps\right)^{k + 1} + \frac{k^2 \log \frac 1 \delta}{\eps^2} \right) 
+ T_{\calL,\calC}^\stat,
\]
for an appropriate constant, in the statistical setting, we can output $\vp:\calX\to \Delta^k$, an $\varepsilon$-omnipredictor for $(\calD,\calL,\calC)$, with probability $\ge 1 - \delta$, given $T^2$ i.i.d.\ samples from $\calD$ and $O(T^{2.5})$ calls to the ERM oracle for each $\calF_i$.
\end{theorem}
\begin{proof}
The proof is largely the same as the proof for Theorem~\ref{thm:glm_multi}. For the online omniprediction result, we combine Lemma~\ref{lem:online_regret_multiclass_mul_general} and Corollary~\ref{cor:online_multi}, and for the statistical omniprediction result, we combine Lemma~\ref{lem:offline_regret_multiclass_mul_general} and Corollary~\ref{cor:offline_multi}. The oracle complexity comes from Lemma~\ref{lem:offline_regret_multiclass_mul_general}.
\end{proof}
\section{Unions of Comparators}\label{sec:union}

In this section, we showcase the flexibility of our framework by applying it to omniprediction against a \emph{union of comparators}. Let $\calL$ be a family of losses $\ell: [-1, 1]^k \times \partial \Delta^k \to [-1, 1]$, and let $\calC\sps i$ be a comparator family satisfying \eqref{eq:cl_multi} for all $i \in [m]$. Our goal is to learn an $(\calL, \calC)$-omnipredictor for
\[\calC \defeq \bigcup_{i \in [m]} \calC_i.\]
In other words, we wish to be competitive against the best $\vc$ in any $\calC_i$. For simplicity, here we focus on the online setting, although similar extensions for statistical omniprediction are straightforward.

To design an online omnipredictor against $(\calL, \calC)$, we define a simultaneous approachability instance as follows: we define $(\calA, \calB)$ as in \eqref{eq:net_setting}, and let
\begin{align*}
\calU^{(i)} \defeq \Brace{\vd_\ell \circ \vc_i}_{\ell \in \calL, \vc \in \calC^{(i)}},\; \vv^{(i)}(\va, \vb) \defeq \E_{\vp \sim \va}[\vp - \vb], \text{ for all } i \in [m], \\
\calU^{(m + 1)} \defeq [-1, 1]^{\calN \times k},\; \vv^{(m + 1)}(\va, \vb) \defeq \Brace{\va_{\vs}(\vs - \vb)}_{\vs \in \calN}.
\end{align*}

In other words, there is one approachability set $\calU\sps i$ for each comparator class $\calC\sps i$, and the $(m + 1)^{\text{th}}$ approachability set is defined analogously to $\calU\sps 1$ in \eqref{eq:uvdef_multi_online}.

\begin{theorem}\label{thm:union_omni}
Let $\calL$ be a family of loss functions and $\calC^{(i)}$ be a family comparators for all $i \in [m]$, such that \eqref{eq:cl_multi} holds for $\calL$ and every $\calC \gets \calC\sps i$, let $\calF\sps i \defeq \{\vd_\ell \circ \vc\}_{\ell \in \calL, \vc \in \calC\sps i}$, and let $\delta \in (0, 1)$. Let $T_{\calL,\calC}^ \seq$ be such that all $T \ge T_{\calL,\calC}^ \seq$ and $i \in [m]$ satisfy
\[\sum_{j \in [k]} \srad_T(\calF_j^{(i)}) \le \frac \eps 9. \]
Then if
\[T = \Omega\left(k\left(\frac 1 \eps\right)^{k + 1} + \frac{\log \frac m \delta}{\eps^2} \right) + T_{\calL,\calC}^ \seq\]
for an appropriate constant, in the online setting, we can output $\vp_{[T]} \in (\Delta^k)^T$, an $\eps$-omnipredictor for $(\vx_{[T]}, \vy_{[T]}, \calL, \bigcup_{i \in [m]} \calC\sps i)$, with probability $\ge 1 - \delta$.
\end{theorem}
\begin{proof}
The proof is almost exactly identical to Theorem~\ref{thm:glm_multi}, save for two changes. First, the additive regret term in Corollary~\ref{cor:context-sba} now scales with $\log(\frac m \delta)$ (as there are $m + 1$ approachability sets). Second, the CMLOO in Lemma~\ref{lem:minimax} now must hold for $m + 1$ inputs. However, when applying Lemma~\ref{lem:minimax} (specifically following the notation \eqref{eq:vvi}), every $\mm^{(i)}$ is identical for $i \in [m]$, and we bounded the quantity \eqref{eq:m_op} for $\mm^{(m + 1)}$ already in Corollary~\ref{cor:online_multi}. Thus, the same proof holds and we simply adjust the logarithmic term in the $T$ lower bound.
\end{proof}

We remark that all of our main results generalize to unions of comparators; indeed, the binary omniprediction CMLOO construction in Lemma~\ref{lem:oracle_binary} also has a simple extension to this setting. Our framework is even capable of handling unions of \emph{loss families} in much the same way, where we define an approachability set to ensure multiaccuracy for each pairing of a loss family and a comparator class, although we omit this extension to avoid tedium.
\newpage
\bibliography{ref.bib}
\bibliographystyle{alpha}

\newpage
\appendix

\section{Deferred Proofs}\label{app:defer}

\restateomni*
\begin{proof}
We begin with the statistical setting. By definition of $\vk^\star_\ell$, for all $\ell \in \calL$ and $\vc \in \calC$,
\begin{equation}\label{eq:kl_better}\E_{\vx \sim \calD_{\vx}}\Brack{\E_{\vy \sim \vp(\vx)}\Brack{\ell\Par{\vk^\star_\ell(\vp(\vx)), \vy}}} \le \E_{\vx \sim \calD_{\vx}}\Brack{\E_{\vy \sim \vp(\vx)}\Brack{\ell\Par{\vc(\vx), \vy}}} ,\end{equation}
and thus
\begin{equation}\label{eq:simulate_arg}
\begin{aligned}
\E_{(\vx, \vy) \sim \calD}\Brack{\ell\Par{\vk^\star_\ell(\vp(\vx)), \vy}} &= \E_{(\vx, \vy) \sim \calD}\Brack{\ell\Par{\vk^\star_\ell(\vp(\vx)), \vy}} - \E_{\vx \sim \calD_{\vx}}\Brack{\E_{\vy \sim \vp(\vx)}\Brack{\ell\Par{\vk^\star_\ell(\vp(\vx)), \vy}}} \\
&+ \E_{\vx \sim \calD_{\vx}}\Brack{\E_{\vy \sim \vp(\vx)}\Brack{\ell\Par{\vk^\star_\ell(\vp(\vx)), \vy}}} - \E_{\vx \sim \calD_{\vx}}\Brack{\E_{\vy \sim \vp(\vx)}\Brack{\ell\Par{\vc(\vx), \vy}}} \\
&+ \E_{\vx \sim \calD_{\vx}}\Brack{\E_{\vy \sim \vp(\vx)}\Brack{\ell\Par{\vc(\vx), \vy}}} - \E_{(\vx, \vy) \sim \calD}\Brack{\ell\Par{\vc(\vx), \vy}}\\
&+ \E_{(\vx, \vy) \sim \calD}\Brack{\ell\Par{\vc(\vx), \vy}} \\
&\le \E_{(\vx, \vy) \sim \calD}\Brack{\ell\Par{\vk^\star_\ell(\vp(\vx)), \vy}} - \E_{\vx \sim \calD_{\vx}}\Brack{\E_{\vy \sim \vp(\vx)}\Brack{\ell\Par{\vk^\star_\ell(\vp(\vx)), \vy}}} \\
&+ \E_{\vx \sim \calD_{\vx}}\Brack{\E_{\vy \sim \vp(\vx)}\Brack{\ell\Par{\vc(\vx), \vy}}} - \E_{(\vx, \vy) \sim \calD}\Brack{\ell\Par{\vc(\vx), \vy}} \\
&+ \E_{(\vx, \vy) \sim \calD}\Brack{\ell\Par{\vc(\vx), \vy}},
\end{aligned}
\end{equation}
where the second line in \eqref{eq:simulate_arg} was bounded by \eqref{eq:kl_better}. Taking an expectation of Lemma~\ref{lem:loss_oi} over $\vx \sim \calD_{\vx}$,
\begin{align*}
\E_{(\vx, \vy) \sim \calD}\Brack{\ell\Par{\vk^\star_\ell(\vp(\vx)), \vy}} - \E_{\vx \sim \calD_{\vx}}\Brack{\E_{\vy \sim \vp(\vx)}\Brack{\ell\Par{\vk^\star_\ell(\vp(\vx)), \vy}}} = \E_{(\vx, \vy) \sim \calD}\Brack{\inprod{\vd_\ell(\vk^\star_\ell(\vp(\vx)))}{\vy - \vp(\vx)}}, \\
\E_{\vx \sim \calD_{\vx}}\Brack{\E_{\vy \sim \vp(\vx)}\Brack{\ell\Par{\vc(\vx), \vy}}} - \E_{(\vx, \vy) \sim \calD}\Brack{\ell\Par{\vc(\vx), \vy}} = \E_{(\vx, \vy) \sim \calD}\Brack{\inprod{\vd_\ell(\vc(\vx))}{\vp(\vx) - \vy}},
\end{align*}
and the conclusion follows by applying Definitions~\ref{def:ma} and~\ref{def:cal}. The online setting is similar:
\begin{align*}
\frac 1 T \sum_{t \in [T]} \ell(\vk^\star_\ell(\vp_t), \vy_t) &= \frac 1 T \sum_{t \in [T]} \Par{\ell(\vk^\star_\ell(\vp_t), \vy_t) -  \E_{\vy \sim \vp_t}\Brack{\ell(\vk^\star_\ell(\vp_t), \vy)}} \\
&+ \frac 1 T \sum_{t \in [T]} \Par{\E_{\vy \sim \vp_t}\Brack{\ell(\vk^\star_\ell(\vp_t), \vy)} - \E_{\vy \sim \vp_t}\Brack{\ell(\vc(\vx_t), \vy)}} \\
&+ \frac 1 T \sum_{t \in [T]} \Par{\E_{\vy \sim \vp_t}\Brack{\ell(\vc(\vx_t), \vy)} - \ell(\vc(\vx_t), \vy_t)} + \frac 1 T \sum_{t \in [T]} \ell(\vc(\vx_t), \vy_t)
\end{align*}
at which point the conclusion again follows from Definitions~\ref{def:ma} and~\ref{def:cal}, because for all $t \in [T]$,
\begin{align*}
\ell(\vk^\star_\ell(\vp_t), \vy_t) -  \E_{\vy \sim \vp_t}\Brack{\ell(\vk^\star_\ell(\vp_t), \vy)} &= \inprod{\vd_\ell(\vk^\star_\ell(\vp_t))}{\vy_t - \vp_t},\\
\E_{\vy \sim \vp_t}\Brack{\ell(\vc(\vx_t), \vy)} - \ell(\vc(\vx_t), \vy_t) &= \inprod{\vd_\ell(\vc(\vx_t))}{\vp_t - \vy_t}.
\end{align*}
\end{proof}
\section{Counterexample for Multiclass Isotonic Regression}\label{sec:lb}

Our framework for multiclass omniprediction was based on the construction of \cite{OkoroaforKK25} in the binary setting. Concurrently, another construction of $\approx \eps^{-2}$-sample complexity binary omnipredictors was given by \cite{hu2025omnipredicting} for GLMs. It is thus natural to ask whether the construction in \cite{hu2025omnipredicting} has a multiclass extension. In this section, we show a barrier to such a generalization.

The \cite{hu2025omnipredicting} construction was based on the \emph{Isotron} algorithm \cite{KalaiS09}, which alternates online gradient descent with isotonic regression. In particular, the isotonic regression problem that Isotron repeatedly solves is, for input labels $\{y_i\}_{i \in [n]}$, and some proper loss function $\ell: [0, 1]^2 \to \R$,
\begin{equation}\label{eq:binary_ir}
\min_{\{p_i\}_{i \in [n]}\in [0,1]^n} \sum_{i \in [n]} \ell(p_i, y_i),
\text{ subject to }  p_i \le p_{i + 1} \text{ for all } i \in [n-1].
\end{equation}
In other words, Isotron finds the best-fitting \emph{monotone} sequence of $\{p_i\}_{i \in [n]}$ with respect to $\{y_i\}_{i \in [n]}$, as measured by $\ell$. The monotonicity requirement comes from $p_i$ being induced by the gradient of a one-dimensional convex function (for more on this relationship, see Lemma~\ref{lem:glm_proper} and Section 2.2, \cite{hu2025omnipredicting}). Crucially, in the binary setting the optimal choice of  $\{p_i\}_{i \in [n]}$ is independent of the choice of proper loss $\ell$ in \eqref{eq:binary_ir} (Corollary 9, \cite{hu2025omnipredicting}; see also \cite{pav-proper}). This omniprediction property of isotonic regression \eqref{eq:binary_ir} is then inherited by the overall Isotron framework.

We next state the natural generalization of \eqref{eq:binary_ir} to the multiclass setting. As \cite{GneitingR07} shows, again any proper loss induces predictions via the gradient of a convex function.
A vector field is the gradient of a convex function iff it is \emph{cyclically monotone} (Theorem B, \cite{Rockafellar70}), and we can capture this high-dimensional condition via the following extension of \eqref{eq:binary_ir}.

\begin{restatable}{problem}{restateblir}\label{prob:blir}
Given $\{(\vv_i, \vy_i)\}_{i \in [T]}\subset \R^k\times \partial \Delta^k$, we define the following isotonic regression problem for a proper loss $\ell: \Delta^k \times \Delta^k \to \R$:
\begin{equation}\label{eq:blir}
\begin{gathered}\{\vp^\star_i, f^\star_i\}_{i \in [n]} \defeq \argmin_{\{\vp_i, f_i\}_{i \in [n]} \in (\Delta^k \times \R)^n}\sum_{i \in [n]} \ell(\vp_i, \vy_i), \\
\text{subject to } \langle \vp_j, \vv_i-\vv_j \rangle\leq f_i-f_j   \text{ for all } (i, j) \in [n] \times [n]. \end{gathered}\end{equation}
\end{restatable}

Here, the $\{\vv_i\}_{i \in [n]}$ should be interpreted as the ``unlinked'' predictors in a GLM, and the monotonicity condition in \eqref{eq:blir} is equivalent to $\vp_i = \nabla \omega(\vv_i)$, $f_i = \omega(\vv_i)$ for all $i \in [n]$. This parameterization is implicit in \eqref{eq:binary_ir} as well, where the input $v_i$ are first sorted to define the indexing.

For the strategy in \cite{hu2025omnipredicting} to generalize to high dimensions, a reasonable necessary condition is for the same omniprediction property to hold for \eqref{eq:blir}, i.e., that its minimizing $\{\vp^\star_i\}_{i \in [n]}$ does not depend on the choice of proper loss $\ell$. We give a simple numerical counterexample. Define:
\[\ell_{\mathrm{sq}}(\vp, \vy) \defeq \half \norm{\vp - \vy}_2^2,\quad \ell_{\mathrm{log}}(\vp,\vy) \defeq -\sum_{i \in [k]} \log(\vp_i) \ind_{\vy = \ve_i}.\]
We minimize \eqref{eq:blir} with respect to these two proper losses, and the following choices of $\{\vv_i, \vy_i\}_{i \in [2]}$:
\begin{equation}\label{eq:bad_vy}
\Brace{\vv_i}_{i \in [2]} =
\begin{bmatrix}
0 & 1 \\
0 & 0 \\
0 & 0 
\end{bmatrix},\quad
\Brace{\vy_i}_{i \in [2]} =
\begin{bmatrix}
1 & 0 \\
0 & 1 \\
0 & 0 
\end{bmatrix}.
\end{equation}

\begin{lemma}
The minimizer of \eqref{eq:blir} with $\ell \gets \ell_{\textup{sq}}$ and inputs \eqref{eq:bad_vy} is
\begin{equation}\label{eq:min_sq}\{\vp_i^\star\}_{i \in [2]} = \begin{bmatrix} 
\frac 3 7 & \frac 3 7 \\ 
\frac 2 7 & \frac 4 7 \\ 
\frac 2 7 & 0
\end{bmatrix}.
\end{equation}
\end{lemma}
\begin{proof}
For $\vv_1 - \vv_2 = -\ve_1$, the constraints in \eqref{eq:blir} are equivalent to
\[[\vp_1]_1 \le [\vp_2]_1.\]
Our goal is to minimize $\norms{\vp_1 - \ve_1}_2^2 + \norms{\vp_2 - \ve_2}_2^2$ subject to this constraint. It is clear that the constraint is tight, because otherwise $\vp_2$ would put any excess mass on the second coordinate. Thus the minimizing $\vp_1$ and $\vp_2$ are of the form
\[\vp_1 = \begin{bmatrix} t \\ \frac{1-t}{2} \\ \frac{1-t}{2} \end{bmatrix},\quad \vp_2 = \begin{bmatrix} t \\ 1 - t \\ 0 \end{bmatrix}.\]
The former claim is because Jensen's inequality implies $\vp_1$ should spread all remaining mass over the last two coordinates, and the latter is because $\vp_2$ has no incentive to place any mass on the third coordinate.
The conclusion follows by solving for $t$ that minimizes $(1 - t)^2 + 2 \cdot \frac 1 4 (1 - t)^2 + 2t^2$.
\end{proof}

\begin{lemma}
The minimizer of \eqref{eq:blir} with $\ell \gets \ell_{\textup{log}}$ and inputs \eqref{eq:bad_vy} is not \eqref{eq:min_sq}.
\end{lemma}
\begin{proof}
It suffices to check that the following choices attain better function value:
\[\vp_1 = \begin{bmatrix} \half \\ \frac 1 4 \\ \frac 1 4 \end{bmatrix}, \quad \vp_2 = \begin{bmatrix} \half \\ \half \\ 0\end{bmatrix}. \]
This is because $-\log(\frac {12}{49}) \ge -\log(\frac{1}{4})$, and the constraint $[\vp_1]_1\le[\vp_2]_1$ is satisfied.
\end{proof}

\end{document}